\documentclass[pra,twocolumn,a4paper,nofootinbib,superscriptaddress]{revtex4-2}

\usepackage{hyperref}
\usepackage{graphicx}
\usepackage{amsmath}
\usepackage{amssymb}
\usepackage{color}
\usepackage{amsthm}
\usepackage{amsfonts}
\usepackage{subfigure}
\usepackage{enumerate}
\usepackage{enumitem}
\usepackage{tikz}
\usepackage{url}
\usetikzlibrary{arrows,snakes}
\newtheorem{theorem}{Theorem}
\newtheorem{lemma}{Lemma}

\newtheorem{corollary}{Corollary}

\begin{document}

\title{Universal Fault-Tolerant Quantum Computing with Stabiliser Codes}
\author{Paul Webster}
\affiliation{Centre for Engineered Quantum Systems, School of Physics, The University of Sydney, Sydney, NSW 2006, Australia}
\author{Michael Vasmer}
\affiliation{Perimeter Institute for Theoretical Physics, Waterloo, ON N2L 2Y5, Canada}
\affiliation{Institute for Quantum Computing, University of Waterloo, Waterloo, ON N2L 3G1, Canada}
\author{Thomas R.~Scruby}
\affiliation{Department of Physics and Astronomy, University College London, London WC1E 6BT, UK}
\author{Stephen D.~Bartlett}
\affiliation{Centre for Engineered Quantum Systems, School of Physics, The University of Sydney, Sydney, NSW 2006, Australia}

\date{May 19, 2021}

\begin{abstract}
The quantum logic gates used in the design of a quantum computer should be both \emph{universal}, meaning arbitrary quantum computations can be performed, and \emph{fault-tolerant}, meaning the gates keep errors from cascading out of control.  A number of no-go theorems constrain the ways in which a set of fault-tolerant logic gates can be universal.  These theorems are very restrictive, and conventional wisdom holds that a universal fault-tolerant logic gate set cannot be implemented natively, requiring us to use costly distillation procedures for quantum computation.  Here, we present a general framework for universal fault-tolerant logic with stabiliser codes, together with a no-go theorem that reveals the very broad conditions constraining such gate sets.  Our theorem applies to a wide range of stabiliser code families, including concatenated codes and conventional topological stabiliser codes such as the surface code.  The broad applicability of our no-go theorem provides a new perspective on how the constraints on universal fault-tolerant gate sets can be overcome. In particular, we show how non-unitary implementations of logic gates provide a general approach to circumvent the no-go theorem, and we present a rich landscape of constructions for logic gate sets that are both universal and fault-tolerant.  That is, rather than restricting what is possible, our no-go theorem provides a signpost to guide us to new, efficient architectures for fault-tolerant quantum computing.
\end{abstract}

\maketitle
\section{Introduction}
Quantum computing offers the promise of efficiently solving problems that are intractable for conventional computers. Delivering on this promise is a challenge, because quantum systems are especially vulnerable to noise, and a quantum advantage can easily be destroyed by decoherence. One approach to this challenge is to use quantum error-correcting codes to construct near-noiseless logical qubits using redundancy in many physical qubits~\cite{Aharanov,Aliferis}. Furthermore, to process the encoded quantum information, we require quantum logic gates that are \emph{fault-tolerant}, meaning that they do not compromise the protection of the code by significantly increasing the probability of logical errors. Realising fault-tolerant implementations of a quantum-computationally universal set of logical operators on a quantum error-correcting code family is essential for scalable quantum computing. 

Constructing a quantum computing architecture with a universal, fault-tolerant gate set has been a longstanding challenge for the field, because the existence of such a gate set for a given quantum code is highly constrained by several no-go theorems \cite{Eastin,Zeng,Bravyi,Pastawski, WebsterLPLO,Beverland,WebsterBraid,JochymOConnor,Burton,Cree}.  The most significant of these is the Eastin-Knill theorem~\cite{Eastin}, which proves that the simplest approach to fault-tolerance -- transversal gates -- cannot yield a universal gate set for any quantum code. For particular code families, more general no-go theorems have subsequently been proven. For example, for topological stabiliser codes, implementing gates by locality-preserving logical operators~\cite{Bravyi} and braiding defects~\cite{WebsterBraid} is insufficient for universality.  The conventional wisdom is that a quantum error correcting code cannot natively possess a universal, fault-tolerant set of logic gates.

Notwithstanding these no-go results, a large and apparently diverse range of approaches have been proposed to implement universal, fault-tolerant quantum logic~\cite{Bravyi2,Bombin3,Bombin,Vasmer,Brown,WebsterBraid,JochymOConnor2,Paetznik,Anderson,Knill2,Yoder2,Yoder,Brown2,Bombin4,JochymOConnor3,Bravyi3,Bombin5,Hwang,Colladay,Campbell}.  The most well-known amongst these approaches is magic state distillation \cite{Bravyi2}, which lowers the logical error rate of a particular non-fault-tolerant gate by brute force, and as such is expensive in terms of resource cost.  With distillation and more broadly, these approaches all make use of some extra technique that lies outside the scope of the no-go theorems.  The theory of quantum computing architectures is left in an undesirable situation, where systematic approaches to designing fault-tolerant quantum logic are highly constrained by no-go theorems, and the `tricks' used to circumvent these constraints are \emph{ad hoc} and lack unifying principles.  The result is a spectrum of potential quantum computing architectures with a mix of efficient and costly elements that are difficult to compare with each other.  

In this paper, we resolve this situation for a broad class of stabiliser code families, by identifying the structures in the sets of logical operators that lead to such no-go theorems as well as the avenues for circumventing the constraints of these theorems.  We present two main results. First (in Sec.~\ref{Sec:III}), we prove a highly general no-go result precluding a universal set of logical operators by quantum channels that are \emph{local noise preserving}, meaning that they preserve \emph{locality} of stochastic noise (in the sense of exponential suppression of high-weight errors, regardless of geometric locality \cite{Gottesman4}). This is significant since local noise preserving channels are those that are naturally fault-tolerant for local stochastic noise \cite{Aharanov,Aliferis,Gottesman4}, generalising standard approaches to fault tolerance such as transversal gates \cite{Eastin} and (geometrically) locality-preserving logical operators \cite{Bravyi}. This theorem applies to a wide range of stabiliser codes that includes concatenated stabiliser codes and topological stabiliser codes with conventional topological order as important subclasses, and can be understood as a substantial generalisation of the Eastin-Knill theorem \cite{Eastin} for these codes. Second (in Sec.~\ref{Sec:IV}), we show how approaches to circumventing the no-go result to realise fault-tolerant implementations of a universal set of logical operators follow naturally from the logic of this theorem. Specifically, we generalise the notion of local noise preserving channels, while retaining fault tolerance, by constraining their action only on correctable errors but allowing any action on uncorrectable errors. This generalisation is trivial for unitary operators, but allows for universal fault tolerance by non-unitary channels that treat logical operators differently from correctable errors. Specifically, we describe two classes of such channels; \emph{gradational} channels, which include code deformation and pieceable fault tolerance, and \emph{intrinsically fault-tolerant} channels which include state injection and code switching. The apparent diversity of approaches to implementing a universal set of logical operators fault-tolerantly thus emerges naturally from our analysis.

\section{Restrictions on Sets of Fault-Tolerant Logical Gates}\label{Sec:III}
{ In this section, we prove our main no-go result. Specifically, we show that a highly general class of fault-tolerant quantum channels are not sufficient for implementing a universal set of logical operators on a large class of stabiliser codes. In Sec.~\ref{Sec:IIA}} we introduce the notion of fault-tolerant quantum channels and argue that the class of \emph{local noise preserving} channels constitutes a natural class of fault-tolerant quantum channels under the assumption of local stochastic noise. In Sec.~\ref{Sec:IIB}, we then prove our no-go result, showing that universal quantum computing using only local noise preserving channels is not possible on a large class of stabiliser codes including conventional topological stabiliser codes and concatenated stabiliser codes. This result clarifies that universal, fault-tolerant quantum computing with such codes necessitates generalising fault tolerance beyond local noise preserving channels, which we consider in Sec.~\ref{Sec:IV}.

\subsection{Fault-Tolerant Quantum Channels}\label{Sec:IIA}
\subsubsection{Quantum Channels}
Operations on quantum error-correcting codes can be described by quantum channels acting on the physical qubits. A quantum channel is a completely-positive trace-preserving (CPTP) map on the state of the qubits. Such a channel admits a dual map, $\mathcal{A}$, that acts on operators on the physical qubits and corresponds to the Heisenberg picture \cite{Attal}. We use this dual map throughout, referring to it as the quantum channel for brevity. We refer to a quantum channel as \emph{unitary} if its action is equivalent to some unitary operator $U$ (i.e.~$\mathcal{U}(B)=UBU^\dag$), but we emphasise that general quantum channels are not unitary. 

\subsubsection{Noise}
We model the noise afflicting a quantum error-correcting code by a quantum channel, $\mathcal{E}_\eta$, that is applied before each timestep of computation. This noise is parameterised by a noise rate $\eta$, independent of the code, which we often leave implicit. In general, this noise can be time-dependent, with a different noise channel, $\mathcal{E}_\eta^{(t)}$ applied at each timestep $t$.

Scalable protection against noise can be achieved with a code family -- a sequence of error-correcting codes indexed by a positive integer $l$. We require that this code family has a threshold, i.e.,~that there exists a non-zero noise rate below which the effect of noise on the logical qubits after error correction can be arbitrarily suppressed by choosing sufficiently large $l$. We quantify  the effect of noise on the logical qubits using the \emph{distance} between the noise channel followed by application of a decoder that performs error-correction, $\mathcal{D}$, i.e.~$\mathcal{D}\circ\mathcal{E}$, and the identity channel, $\mathcal{I}$. We follow Ref.~\cite{Aharanov} by using the diamond norm as as our measure of distance, and restrict to the logical space, since we only require that error correction succeeds if the original state is a logical state. We denote this distance by $\| \cdot \|_{\diamond,\mathcal{L}}$. A code family (with decoder, $\mathcal{D}$) then has a threshold, $\eta_{\mathcal{I}}>0$, if
\begin{equation}\label{Threshold}
  \lim_{l\to\infty} \left\|\mathcal{D}\circ\mathcal{E}_\eta-\mathcal{I}\right\|_{\diamond,\mathcal{L}}=0\,, \quad \text{for}\ \eta<\eta_{\mathcal{I}}\,.
\end{equation}
In the case of time-dependent noise, we require that $\mathcal{E}^{(t)}$ admits a threshold for all $t$.

\subsubsection{Fault Tolerance}
We now consider fault tolerance of quantum logic gates, expressed as quantum channels. Precisely, a \emph{logical operator implementation} is a quantum channel that acts as a logical operator on the logical space of a code. Specifically, an implementation of logical operator $\bar{A}$ is a quantum channel such that, for any logical operator $\bar{U}$, if $U$ is a representative of $\bar{U}$, then $\mathcal{A}(U)$ is a representative of logical operator $\bar{A}\bar{U}\bar{A}^\dag$. A logical operator implementation can have any action on non-logical operators. 

In the presence of noise, attempting to apply a quantum channel $\mathcal{A}$ will instead result in a noisy implementation $\mathcal{A}_{\mathcal{E}_\eta}$. If the time over which $\mathcal{A}$ is applied is constant in the code size then we can assume that $\mathcal{A}_{\mathcal{E}_\eta}=\mathcal{A} \circ \mathcal{E}_\eta$; more generally, $\mathcal{A}_{\mathcal{E}_\eta}$ corresponds to $\mathcal{A}$ along with the application of $\mathcal{E}^{(t)}$ before each timestep, $t$, of its implementation.

We define $\mathcal{A}$ to be \emph{fault-tolerant} if $\mathcal{A}_{\mathcal{E}_\eta}$ followed by decoding approximates $\mathcal{A}$ on the logical space to arbitrarily high precision. Precisely, we define $\mathcal{A}$ to be fault-tolerant (with decoder, $\mathcal{D}$) if there exists a threshold $\eta_{\mathcal{A}}>0$ such that
\begin{equation}\label{EqFT2}
  \lim_{l\to\infty} \left\|\mathcal{D}\circ\mathcal{A}_{\mathcal{E}_{\eta}}-\mathcal{A}\right\|_{\diamond,\mathcal{L}}=0\,, \quad \text{for}\ \eta<\eta_{\mathcal{A}}\,.
\end{equation}

Thus, to directly verify that a set of quantum channels, $\{\mathcal{A}\}$, can be used to implement fault-tolerant quantum computing within a given architecture, we must first specify a noise model, $\mathcal{E}$, and decoder, $\mathcal{D}$, and then show that Eq.~(\ref{EqFT2}) is satisfied for each $\mathcal{A}$. Such a direct approach to constructing fault-tolerant gate sets can be taken (for example, see Ref.~\cite{JochymOConnor2} or Appendix \ref{AppendixE}), but in general a more flexible approach is preferred whereby fault tolerance of a class of quantum channels is ensured for a wide range of realistic noise models and decoders. We now introduce such a naturally fault-tolerant class of channels; those that are local noise preserving.

\subsubsection{Local Noise Preserving Channels}
A standard class of noise models considered in the study of fault tolerance is those that are \emph{local} and \emph{stochastic} \cite{Aliferis,Gottesman4}. We define a stochastic noise model by specifying a set of unitary errors, $\mathbf{E}$, and an associated probability distribution, $\{p_i\}$, such that $\mathcal{E}$ acts by the application of an error $E_i\in\mathbf{E}$ chosen with probability $p_i$, independent of the current state of the code (i.e.~of past errors). A stochastic noise model is \emph{local} if the probability that any given set of physical qubits, $\textbf{A}$, is contained in the support of an error decays exponentially with the size of $\textbf{A}$, i.e.~ $p\left(\{E_i|\textbf{A}\subseteq \text{supp}(E_i)\}\right) \leq  \eta^{|\textbf{A}|}$ where $\eta<1$ is the noise rate \cite{Gottesman4}. We emphasise that this does not imply any assumptions on the \emph{geometric} locality of noise.

The action of a quantum channel, $\mathcal{A}$ can be understood to map a stochastic noise model, $\mathcal{E}$ to a different, \emph{effective} stochastic noise model, $\mathcal{E}_{\mathcal{A}}$. Informally, a quantum channel should be fault-tolerant if this mapping preserves properties of noise channels that make them amenable to error correction. Precisely, for quantum channel $\mathcal{A}$ and stochastic noise channel $\mathcal{E}$, we can write $\mathcal{A}_{\mathcal{E}_\eta}= \mathcal{E}_{\mathcal{A}_\eta} \circ \mathcal{A}$ where $\mathcal{E}_{\mathcal{A}}$ is some stochastic noise model. Indeed, for constant-time $\mathcal{A}$, $\mathcal{E}_{\mathcal{A}_\eta}$ corresponds to the stochastic noise model with probability distribution $p_{\mathcal{A}}(E)=\sum_{E'\in \mathcal{A}^{-1}(E)}p(E')$; more generally $p_{\mathcal{A}}(E)$  is the sum over the probabilities of all sets of errors occuring over the application of $\mathcal{A}$ that result in the final error, $E$. Thus, we can consider the effect of applying $\mathcal{A}_{\mathcal{E}_\eta}$ to be equivalent to applying the ideal channel $\mathcal{A}$ followed by an effective noise channel $\mathcal{E}_{\mathcal{A}_\eta}$. 

A logical operator implementation, $\mathcal{A}$, is thus fault-tolerant if the noise model $\mathcal{E}_{\mathcal{A}_\eta}$ necessarily admits a threshold provided that $\mathcal{E}^{(t)}$ admits a threshold for all $t$. Indeed, Eq.~(\ref{EqFT2}) is equivalent to
\begin{equation}\label{EqFT3}
  \lim_{l\to\infty} \left\|\left(\mathcal{D}\circ\mathcal{E}_{\mathcal{A}_{\eta}}-\mathcal{I}\right)\circ \mathcal{A}\right\|_{\diamond,\mathcal{L}}=0\,, \quad \text{for}\ \eta<\eta_{\mathcal{A}}\,.
\end{equation}
As a logical operator implementation, the action of $\mathcal{A}$ on the logical space is unitary and, hence, invertible. Thus,  from Eq.~(\ref{EqFT3}), we have that fault tolerance of $\mathcal{A}$ is equivalent to the condition that there is a threshold for noise model $\mathcal{E}_{\mathcal{A}}$, i.e.
\begin{equation}\label{EqFT4}
  \lim_{l\to\infty} \left\|\mathcal{D}\circ\mathcal{E}_{\mathcal{A}_{\eta}}-\mathcal{I}\right\|_{\diamond,\mathcal{L}}=0\,, \quad \text{for}\ \eta<\eta_{\mathcal{A}}\,.
\end{equation}

Theorems exist showing the existence of a threshold for local stochastic noise for a wide range of code families, including concatenated codes \cite{Aharanov} and low-density parity check (LDPC) codes (e.g.~topological stabiliser codes) \cite{Kovalev,Gottesman4}, with appropriate decoders. Provided that $\mathcal{E}_{\mathcal{A}}$ is a local stochastic noise model, these theorems imply,  via Eq.~(\ref{EqFT4}), that quantum channel $\mathcal{A}$ is fault-tolerant for such decoders.

Motivated by this observation, we introduce the class of \emph{local noise preserving (LNP) quantum channels}, as quantum channels that map local stochastic noise models to local stochastic noise models. Precisely, we define quantum channel $\mathcal{A}$ to be LNP if, for any set of local stochastic noise models $\{\mathcal{E}^{(t)}\}$, the corresponding $\mathcal{E}_{\mathcal{A}}$ is a local stochastic noise model. LNP quantum channels on a code family thereby serve as a set of quantum channels that are manifestly fault-tolerant for any local stochastic noise model, and any decoder that gives a threshold for local stochastic noise. They can be understood as a generalisation of important classes of fault-tolerant gates, including transversal gates \cite{Eastin} and (geometrically) locality-preserving logical operators \cite{Bravyi} (as shown in Appendix \ref{AppendixC}), which are based on the principle of preserving locality of errors. By contrast, as shown in the following subsection, quantum channels that spread errors by a factor unbounded in the code size are not LNP.

\subsection{No Go Theorem}\label{Sec:IIB}

\subsubsection{Spread of Errors}
The fault tolerance of quantum channels is often associated with the extent to which they spread errors. We justify this association by showing that LNP quantum channels indeed can only spread a Pauli error by a factor at most constant in the code size.

Specifically, we define the \emph{spread} of a quantum channel as a measure of the extent to which it can increase the size of the support of errors. More precisely, the spread of a Pauli operator $P\in\mathbf{P}_n$ {under the action of a quantum channel} $\mathcal{A}$ is
\begin{equation}
s_{\mathcal{A}}\left(P\right)= \frac{|\text{supp}\left(\mathcal{A}(P)\right)|}{|\text{supp}\left(P\right)|}
\end{equation}
We define the spread of $\mathcal{A}$ to be the maximum spread of any Pauli operator under $\mathcal{A}$:
\begin{equation}
s_{\mathcal{A}}=\max_{P\in\mathbf{P}_n} s_{\mathcal{A}}(P)
\end{equation}
We say $\mathcal{A}$ has bounded spread if there exists a constant $C$ (independent of $l$) such that $s_{\mathcal{A}}\leq C$ for all $l\in\mathbb{Z}_+$.

Bounded-spread quantum channels do not significantly increase the size of the support of errors. Many of the standard approaches to fault tolerance make use of bounded-spread quantum channels (see Appendix~\ref{AppendixC} for examples). Indeed, transversal gates have spread that is bounded by the number of code blocks they act on, and so have bounded spread (assuming they act on a fixed number of code blocks). More generally, the spread of a quantum circuit consisting of local unitary gates scales with its depth, and so constant-depth quantum circuits also have bounded spread. Locality-preserving logical operators in topological codes can only grow the support of errors within a constant-size neighbourhood around the original support \cite{Bravyi,Beverland,WebsterLPLO}, and so also have bounded spread. Finally, braiding defects in topological stabiliser codes is also bounded-spread, because such braiding operations can only increase the weight of logical operators (or uncorrectable errors) by a constant factor (as shown in Appendix \ref{AppendixC3}) and correctable errors are corrected before they can spread beyond a local region.

The following theorem relates LNP quantum channels to bounded-spread quantum channels.
\begin{theorem}\label{Th1}
An LNP quantum channel necessarily has bounded spread.
\end{theorem}
\begin{proof}
Let $\mathcal{A}$ be a quantum channel. By definition of the spread of $\mathcal{A}$, there exists $P\in\mathbf{P}_n$ such that $s_{\mathcal{A}}= \frac{|\text{supp}\left(\mathcal{A}(P)\right)|}{|\text{supp}\left(P\right)|}$. We define a stochastic noise model, $\mathcal{E}_\eta$, that acts as $P$ with probability $p(P)=\eta^{|\text{supp}(P)|}$ or as the identity with probability $p(I)=1-\eta^{|\text{supp}(P)|}$. This is a local stochastic noise model, since it satisfies $p(P)\leq \eta^{|\text{supp}(P)|}$. We consider a time-dependent local stochastic noise model that acts as $\mathcal{E}_\eta$ at the first timestep, and trivially at all subsequent timesteps. 

With this setup, we show that the effective noise model $\mathcal{E}_{\mathcal{A}_\eta}$ is local only if $\mathcal{A}$ has bounded spread, which implies the theorem. Indeed, with noise model $\mathcal{E}_{\mathcal{A}}$ the error $\mathcal{A}(P)$ occurs with probability $p(P)$, since it is present after the implementation of $\mathcal{A}$ iff the error $P$ occurs at the first timestep. Thus, $\mathcal{E}_{\mathcal{A}}$ is a local stochastic noise model only if there exists a constant effective noise rate $\eta_{\mathcal{A}}\in(0,1)$ (independent of $l$) such that 
\begin{equation}
p(P)\equiv \eta^{|\text{supp}(P)|} \leq \eta_{\mathcal{A}}^{|\text{supp}\left(\mathcal{A}(P)\right)|} \,\,\,\,\forall l\in\mathbb{Z}_+
\end{equation}
which implies that
\begin{equation}
s_{\mathcal{A}}\equiv \frac{|\text{supp}\left(\mathcal{A}(P)\right)|}{|\text{supp}\left(P\right)|} \leq \frac{\log(\eta)}{\log\left(\eta_{\mathcal{A}}\right)}\,\,\,\,\,\,\forall l\in\mathbb{Z}_+ \label{EqProof}
\end{equation}
Since, $\eta$ and $\eta_{\mathcal{A}}$ are independent of $l$, it follows from Eq.~(\ref{EqProof}) that $\mathcal{A}$ indeed necessarily has bounded spread.

\end{proof}

\subsubsection{No-Go for Bounded-Spread Quantum Channels}\label{Sec:IIIB}
We here present a no-go theorem that rules out the existence of a universal set of logical operators that admit bounded-spread implementations, that is applicable to a broad class of stabiliser code families.  This class of stabiliser codes includes important code families such as conventional topological stabiliser codes and concatenated stabiliser codes. Our result is important in its own terms, but also takes on particular significance in the context of Theorem \ref{Th1}, because it thereby rules out a universal set of logical operators admitting LNP implementations.

The no-go theorem stems from the observation that bounded-spread logical operator implementations cannot significantly increase the size of the support of logical Pauli operators of stabiliser codes, and this constrains the set of such logic gates. Specifically, let $d_{\bar{L}}$ denote the distance of logical operator $\bar{L}$, i.e.,~the minimum size of support of a representative of $\bar{L}$. If $\mathcal{A}$ is an implementation of logical operator $\bar{A}$ then it must map representatives of logical operator $\bar{L}$ to representatives of $\bar{A}\bar{L}\bar{A}^\dag$.  Because representatives of logical Pauli operators of stabiliser codes are physical Pauli operators \cite{Gottesman2}, this implies that
\begin{equation}
s_{\mathcal{A}}\geq \frac{d_{\bar{A}\bar{P}\bar{A}^\dag}}{d_{\bar{P}}}
\end{equation}
for all logical Pauli operators, $\bar{P}$, where $s_{\mathcal{A}}$ is the spread of $\mathcal{A}$.
Thus, $\bar{A}$ can only admit a bounded-spread implementation if 
\begin{equation}\label{Eq13}
\frac{d_{\bar{A}\bar{P}\bar{A}^\dag}}{d_{\bar{P}}}\leq C
\end{equation}
for some constant $C$ independent of $l$ for all logical Pauli operators, $\bar{P}$. Thus, if there is no universal set of logical operators, $\{\bar{A}\}$, that satisfy Eq.~(\ref{Eq13}) on a given code family, then the code family cannot admit a universal set of bounded-spread logical operator implementations.

To formalise this argument, we define the subset of logical operators, $\mathbf{L}_\downarrow \subseteq \mathbf{L}$, of distance at most a constant factor larger than the code distance $d$, i.e.,
\begin{equation}
\mathbf{L}_\downarrow = \left\{\bar{L}\in\mathbf{L}\middle|\exists C\in\mathbb{R}: \frac{d_{\bar{L}}}{d}\leq C \text{ for all }l\in\mathbb{Z}_+\right\}\,.
\end{equation}
We also define $\mathbf{P}_\downarrow=\mathbf{L}_\downarrow \cap \mathbf{P}_n$ to be the subset of logical Pauli operators in $\mathbf{L}_\downarrow$. The following lemma shows that a logical operator that admits a bounded-spread implementation necessarily maps elements of $\mathbf{P}_\downarrow$ to elements of $\mathbf{L}_\downarrow$.
\begin{lemma}\label{LemmaPL}
Let $\mathcal{A}$ be a bounded-spread implementation of logical operator $\bar{A}$. Then
\begin{equation}
\bar{P}\in\mathbf{P}_\downarrow\implies \bar{A}\bar{P}\bar{A}^\dag \in \mathbf{L}_\downarrow \,.
\end{equation}
\end{lemma}
\begin{proof}
If $\bar{P} \in \mathbf{P}_\downarrow$ then there exists a constant $C$ such that $d_{\bar{P}}/d\leq C$. Thus,
\begin{equation}
s_{\mathcal{A}}\geq \frac{d_{\bar{A}\bar{P}\bar{A}^\dag}}{d_{\bar{P}}} = \frac{d_{\bar{A}\bar{P}\bar{A}^\dag}/d}{d_{\bar{P}}/d}\geq C^{-1} \frac{d_{\bar{A}\bar{P}\bar{A}^\dag}}{d} \,.
\end{equation}
Thus, $\mathcal{A}$ can have bounded spread only if $d_{\bar{A}\bar{P}\bar{A}^\dag}/d$ is bounded, which implies that $\bar{A}\bar{P}\bar{A}^\dag\in\mathbf{L}_\downarrow$.
\end{proof}

We define a code family to be $\mathbf{B}$-\emph{constrained} if the group 
\begin{equation}
\mathbf{B}=\left\langle\left\{\bar{B}\in\mathbf{L}\middle| \bar{B}\bar{P}\bar{B}^\dag \in \mathbf{L}_\downarrow \,\, \forall \bar{P}\in\mathbf{P}_\downarrow \right\}\right\rangle
\end{equation}
is not universal. The following no-go theorem then follows.
\begin{theorem}\label{Th2}
A $\mathbf{B}$-constrained code family cannot admit bounded-spread implementations of a universal set of logical operators.
\end{theorem}
\begin{proof}
By Lemma \ref{LemmaPL}, the set of logical operators that admit bounded-spread implementations is contained in the set $\left\{\bar{B}\in\mathbf{L}\middle| \bar{B}\bar{P}\bar{B}^\dag \in \mathbf{L}_\downarrow \,\, \forall \bar{P}\in\mathbf{P}_\downarrow \right\}$. Thus, if the group generated by this set is not universal then the set of logical operators that implement bounded-spread implementations cannot be universal.
\end{proof}

This theorem implies that any approach to fault tolerance based on bounded-spread quantum channels cannot be universal in $\mathbf{B}$-constrained code families. Indeed, Theorem \ref{Th2} implies the Eastin-Knill theorem \cite{Eastin} for $\mathbf{B}$-constrained code families and the Bravyi-K{\"o}nig theorem for conventional topological stabiliser codes \cite{Bravyi} (which we show below are $\mathbf{B}$-constrained). This is also true for non-unitary, bounded-spread implementations, such as those by braiding defects. In particular, Theorem 1 of Ref.~\cite{WebsterBraid} is also implied (for the case of defects with Pauli stabilisers).

From Theorem \ref{Th1} and Theorem \ref{Th2} we hence have the following corollary.
\begin{corollary}\label{Cor1}
A universal set of logical operators on a $\mathbf{B}$-constrained code family cannot be implemented by local noise preserving quantum channels.
\end{corollary}
Corollary \ref{Cor1} is our main no-go result. It shows that the most naturally fault-tolerant quantum channels for local stochasitic noise -- those that are local noise preserving -- are not sufficient for universal quantum computing in $\mathbf{B}$-constrained code families. This is significant because large and important classes of stabiliser code families are $\mathbf{B}$-constrained, as we now show.

\subsubsection{$\mathbf{B}$-constrained stabiliser code families}

Stabiliser codes are widely used in quantum computing architectures, and many of the most commonly-considered stabiliser code families are $\mathbf{B}$-constrained.  Here, we show that codes that are asymmetric or have infinite disjointness are necessarily $\mathbf{B}$-constrained, which implies that important classes of code families such as conventional topological stabiliser codes and concatenated stabiliser codes are $\mathbf{B}$-constrained.

Asymmetric code families are those with significantly different distances of logical Pauli operators. Precisely, a code family is symmetric if it admits a set of generators, $\mathbf{P}_g$ of the logical Pauli group such that $d_{\bar{P}}\propto d$ for all $\bar{P}\in\mathbf{P}_g$. Otherwise, it is asymmetric. Asymmetric code families are expected to be $\mathbf{B}$-constrained, as their logical Pauli operators can be partitioned into sets of lower-distance operators and higher-distance operators which cannot be interchanged by bounded-spread quantum channels.  This intuition can be made precise:  as shown in Appendix \ref{AppendixSecIIIB}, $\mathbf{L}_\downarrow$ for asymmetric code families is contained in the subspace of the logical operators spanned by $\mathbf{P}_\downarrow$, so $\mathbf{B}$ leaves invariant the subspace stabilised by these logical Pauli operators, precluding universality. Thus, asymmetric stabiliser code families are constrained by Theorem~\ref{Th2} and Corollary~\ref{Cor1}.

An illustrative example of an asymmetric code family is the three-dimensional surface code \cite{Vasmer}. This has logical $\bar{X}$ and $\bar{Z}$ operators supported on regions of two and one dimensions respectively, and so their ratio scales with the linear dimensions of the code making the code family asymmeteric. For this code family, $\mathbf{L}_\downarrow=\{\bar{I},\bar{Z}\}$ (since $\bar{Z}$ is the only one-dimensional logical operator), so $|\bar{0}\rangle$ is stabilised by all bounded-spread logical operator implementations. This precludes a universal set of bounded-spread logical operator implementations and hence a universal set of LNP logical operator implementations.

An alternative way for a code family to be $\mathbf{B}$-constrained is for the group generated by $\mathbf{L}_\downarrow$ to be the logical Pauli group. This is sufficient because then $\mathbf{B}$ is contained in the Clifford group, which is not universal. A sufficient condition for this is that the code family is symmetric and has infinite disjointness. The disjointness of a stabiliser code is a measure of how many almost disjoint representatives there are of each logical Pauli operator (see Appendix \ref{AppendixB} or Ref.~\cite{JochymOConnor} for a more precise definition) and infinite-disjointness code families are those with a disjointness, $\Delta$, satisfying $\lim_{l\to\infty} \Delta(l)=\infty$. As shown in Appendix \ref{AppendixSecIIIB}, for infinite-disjointness code families all logical non-Pauli operators, $\bar{U}$, satisfy $\lim_{l\to\infty}\frac{d_{\bar{U}}}{d_{\bar{P}}}=\infty$ for some logical Pauli operator $\bar{P}$. For a code family that is also symmetric, this implies that a generating set of logical Pauli operators is contained in $\mathbf{P}_\downarrow$ while no logical non-Pauli operators are in $\mathbf{L}_\downarrow$, so $\mathbf{B}$ is contained in the Clifford group. Thus, symmetric, infinite-disjointness stabiliser code families are constrained by Theorem \ref{Th2} and Corollary \ref{Cor1}.

An illustrative example of a symmetric, infinite-disjointness code family is the two-dimensional surface code. This has only string-like logical Pauli operators which all have the same distance, $l$, so it is symmetric. By contrast, all non-Pauli logical operators, such as the Hadamard gate, have support on the full code and so have distance $\Theta(l^2)$. Thus, we indeed have that $\mathbf{L}_\downarrow$ is the logical Pauli group and it is clear that logical Pauli operators cannot be mapped to logical non-Pauli operators by bounded-spread quantum channels, so the bounded-spread logical operator implementations are indeed in the Clifford group and so not universal.

Importantly, the two cases -- asymmetric and symmetric with infinite disjointness -- taken together imply that all stabiliser code families with infinite disjointness are $\mathbf{B}$-constrained. Thus, the no-go results apply to all infinite-disjointness code families which, as we show in Appendix \ref{AppendixB} includes a number of important classes of stabiliser code families. Indeed, all concatenated stabiliser code families have infinite disjointness (see Appendix \ref{AppendixB1}). Conventional topological stabiliser codes -- i.e.,~those with conventional topological order~\cite{Dua} -- also all have infinite disjointness (see Appendix \ref{AppendixB2}). These are topological stabiliser codes that can be associated with topological quantum field theories (TQFTs) with excitations that correspond to error syndromes and which can propagate freely throughout the code \cite{Dua}. This includes all two-dimensional topological stabiliser codes \cite{Bombin2} and all scale and translationally symmetric (STS) codes \cite{Yoshida}, including the surface and colour codes in all spatial dimensions. It remains an open question whether more general classes of LDPC codes also necessarily have infinite disjointness and are therefore also constrained by our no-go results.

We conclude this section by noting that our results for stabiliser codes straightforwardly generalise to subsystem codes. Specifically, for subsystem codes we can define distances and disjointness to be calculated only on the bare logical operators. Considering only the bare logical operators is sufficient because logical operator implementations necessarily map bare logical operators to bare logical operators \cite{Pastawski}. Hence, the proofs of all the results can naturally generalise to the case of subsystem codes by considering bare logical operators and so Theorems \ref{Th1} and \ref{Th2} and Corollary \ref{Cor1} also apply to subsystem codes.

\section{Universal Fault-Tolerant Logical Operator Implementations}\label{Sec:IV}
\subsection{Circumventing Corollary \ref{Cor1}}\label{Sec:IVAA}
The results of Sec.~\ref{Sec:III} place constraints on sets of fault-tolerant logical operators. Importantly, however, these no-go theorems also point to the ways in which these constraints can be circumvented. Specifically, we observe that the logical action of a quantum channel depends only on its action on logical operators, which are necessarily not correctable errors. However, it is sufficient for fault tolerance that a quantum channel maps correctable errors to correctable errors, regardless of its action on uncorrectable errors. Corollary \ref{Cor1} can thus be circumvented by constructing fault-tolerant quantum channels that have bounded spread on correctable errors -- so that correctable errors are not spread to uncorrectable errors -- but unbounded spread on uncorrectable errors.

Specifically, we generalise the class of local noise preserving quantum channels to \emph{essentially local noise preserving} (\emph{ELNP}) quantum channels. Informally, these are quantum channels that preserve the locality of stochastic noise models up to a set of errors of vanishing probability.  Precisely, we define two stochastic noise models, $\mathcal{E}$ and $\hat{\mathcal{E}}$ (with respective probability distributions $p$ and $\hat{p}$) to be \emph{essentially equivalent} if there exists a set of errors $\mathbf{E}$ such that $\lim_{l\to\infty} p(\mathbf{E})=1$ and $p(E)=\hat{p}(E)$ for all $E\in\mathbf{E}$. We define a stochastic noise model to be \emph{essentially local} if it is essentially equivalent to some local stochastic noise model.  We define quantum channel $\mathcal{A}$ to be \emph{essentially local noise preserving (ELNP)} if, for any set of local stochastic noise models $\{\mathcal{E}^{(t)}\}$, the corresponding $\mathcal{E}_{\mathcal{A}}$ is an essentially local stochastic noise model.

ELNP quantum channels straightforwardly inherit fault tolerance for local stochastic noise from LNP quantum channels. Indeed, if $\mathcal{E}$ and $\hat{\mathcal{E}}$ are equivalent noise models, then $\lim_{l\to\infty} \|\mathcal{E}-\hat{\mathcal{E}}\|=0$. Thus (by the triangle inequality on Eq.(~\ref{Threshold})), $\hat{\mathcal{E}}$ admits a threshold for a given code family and decoder if $\mathcal{E}$ does. Hence, threshold theorems for local stochastic noise trivially extend to essentially local stochastic noise models, and so imply the fault tolerance of ELNP quantum channels.

Unitary quantum channels that are ELNP still necessarily have bounded spread and hence cannot allow circumvention of Corollary \ref{Cor1} (as proven in Appendix \ref{AppendixA}). This is because for a unitary quantum channel, $\mathcal{U}(E)=UEU^\dag$, the spread is lower-bounded by the spread of some single-qubit error. This fact holds because any Pauli can be decomposed into a product of single-qubit Pauli operators, $P=\prod_i P_i$, and for a unitary operator, $U$, we have that $U\left(\prod_i P_i \right)U^\dag=\prod_i UP_iU^\dag$. For any single-qubit error, $P_i$, there exists a local stochastic noise model where the probability of $P_i$ is constant in $l$. Thus, an unbounded-spread unitary quantum channel necessarily results in an error with unbounded support size and hence cannot be ELNP.

Importantly, however, there exist non-unitary quantum channels that are ELNP but not LNP. Indeed, if $\mathcal{A}$ is a non-unitary quantum channel, then it does not necessarily follow that $\mathcal{A}(\prod_i P_i )=\prod_i \mathcal{A}(P_i)$. In particular, it is possible for a non-unitary quantum channel to have bounded spread on all correctable errors for some decoder but unbounded spread on a set of logical operators. Provided the decoder gives a threshold on all local stochastic noise models, the set of correctable errors necessarily asymptotically has probability one, and hence such a quantum channel is ELNP despite having unbounded spread. Non-unitary quantum channels of this type can be constructed by integrating error-correction, which by definition corrects correctable errors (thereby preventing their spread) while leaving logical operators unchanged. As we show in this section, a diverse range of techniques for realising a universal set of fault-tolerant logical operator implementations can be understood as instances of ELNP quantum channels of unbounded-spread constructed in this way.

The challenge in realising a logical operator implementation that integrates error-correction into its action in this way is that the two processes cannot be separated. Indeed, it is not sufficient to simply first perform error-correction and then the logical operator implementation, since new errors can arise between these processes that will not be corrected before the logical operator implementation acts on them. There are two ways, however, that the necessary integration can be achieved. 

One way is to decompose an unbounded-spread logical operator implementation into an unbounded-length sequence of bounded-spread logical operator implementations. This decomposition can then be interspersed with conventional error-correction which removes correctable errors before they are allowed to grow too much under the action of the logical operator implementation. We refer to such channels as \emph{gradational} and consider them in Sec.~\ref{Sec:IVA}.

The second way is to use an unconventional approach to error-correction that can be integrated more directly into a logical operator implementation. Specifically, we note that error correction by teleportation (as described in Ref.~\cite{Knill}) allows for simultaneous implementation of a logical operator in a way that yields the teleportation gadget of Ref.~\cite{Gottesman}. Arguably, this idea is what underpins standard universal, fault-tolerant schemes, of which approaches such as state injection and code switching are variants. We refer to such channels as \emph{intrinsically fault-tolerant} and consider them in Sec.~\ref{Sec:IVB}.

Gradational and intrinsically fault-tolerant quantum channels represent two broad classes of fault-tolerant logical operator implementations that can circumvent our no-go results. Table \ref{Table1} summarises their manifestations in a range of approaches to fault-tolerantly implementing a universal set of logical operators.

\begin{table*}
{
  \begin{tabular}{| c | c | c |}
    \hline
     & \textbf{Code Switching} & \textbf{Direct Implementation} \\ \hline
    \textbf{Gradational} & Code Conversion \cite{Hwang,Colladay} & Pieceable Fault Tolerance \cite{Knill2,Yoder2}   \\ 
& Dimensional Conversion & with Delayed Error Correction \cite{Brown}\\
\hline
    \textbf{Intrinsically Fault-Tolerant} & Gauge fixing \cite{Bombin,Paetznik,Anderson} & Magic State Injection \cite{Campbell} \\
 & Dimensional Jumping \cite{Bombin5} & Coherent State Injection \cite{WebsterBraid,Vasmer,Bombin3,Yoder}\\
\hline
  \end{tabular}}
\caption{Summary of approaches to fault-tolerant, universal gate sets. The approaches are classified as gradational or intrinsically fault-tolerant, corresponding to the two classes presented in Sec.~\ref{Sec:IVA} and Sec.~\ref{Sec:IVB} respectively. They are also classified based on whether they are a form of code switching (fault-tolerant mapping between two codes) or direct implementation (of an unbounded-spread logical operator within a single code).}
\label{Table1}
\end{table*}

\subsection{Gradational Logical Operator Implementations}\label{Sec:IVA}
A gradational logical operator implementation is one of the form
\begin{equation}\label{IU3}
\mathcal{A}=\prod_{i=1}^t \mathcal{K}_i\mathcal{U}_i
\end{equation}
where each $\mathcal{U}_i$ is a unitary, LNP (and hence bounded-spread) quantum channel and each $\mathcal{K}_i$ is a quantum channel that allows for error-correction, but acts trivially on logical operators. The action of such an implementation consists of channels $\mathcal{U}_i$ that {each fault-tolerantly} transform the logical space interspersed with channels $\mathcal{K}_i$ that are used to remove errors to ensure that the state of the code remains in the appropriate logical space at all times. Gradational logical operator implementations are thus fault-tolerant by construction. Indeed, the non-unitary channels, $\mathcal{K}_i$,  ensure they are ELNP by preventing the spread of correctable errors, but also allow the spread of uncorrectable errors necessary to circumvent Corollary \ref{Cor1}. 

Gradational logical operator implementations can be understood as the minimal approach to circumventing Corollary \ref{Cor1}. Indeed, gradational logical operator implementations can be decomposed into unitary quantum channels used to transform the logical information interspersed with error correction, which is in the spirit of standard fault-tolerant quantum computing. However, unlike unitary implementations of logical operators, intermediate error-correction is integrated into gradational logical operator implementations themselves, and cannot be avoided by increasing the code size. Indeed, since the spread of each timestep is bounded (as each $\mathcal{U}_i$ must have bounded spread), the time taken for an unbounded-spread, gradational logical operator implementation necessarily grows with the code size, $l$. Thus, gradational logical operator implementations allow for universal fault-tolerant quantum computing in a natural way, but only at the cost of introducing a significant time overhead associated with unbounded-spread logical operator implementations.

We now consider two common approaches to fault-tolerant logic gates that are specific examples of gradational logical operator implementations:  code deformation, and pieceable fault tolerance.

\subsubsection{Code Deformation}\label{Sec:IVA1}
An important class of gradational logical operator implementations is \emph{code deformation}.  A gradational implementation of logical operator $\bar{A}$ can be constructed by decomposing a unitary implementation of $\bar{A}$ into single and two-qubit gates.  By staggering the implementation of these gates such that only a constant number are implemented in each timestep, we can ensure that the spread at each such timestep remains bounded. The product of gates performed at the $i$th timestep then corresponds to a bounded-spread unitary quantum channel $\mathcal{U}_i$.  As a consequence of this staggering, each $\mathcal{U}_i$ is not individually a logical operator implementation, and in general will not preserve the logical space.  Rather, each $\mathcal{U}_i$ will map the logical code space to the logical code space of a different code. This process can be understood as a form of code deformation, $\mathcal{C}\to\mathcal{C}_1\to\cdots\to\mathcal{C}_{m-1}\to\mathcal{C}$. This code deformation can be interspersed with error correction performed at the $i$th timestep by a decoder, $\mathcal{D}_i$ on code $\mathcal{C}_i,$ to construct a gradational implementation of logical operator $\bar{A}$ of the form:
\begin{equation}
\mathcal{A}=\prod_{i=1}^t \mathcal{D}_i\mathcal{U}_i \,.
\end{equation}
We note that it is standard to instead present code deformation as a measurement-based operation, in which the unitaries are omitted and measurements used to simultaneously project to a new codespace and extract an error syndrome. Our presentation is equivalent, but separates the two processes to emphasise how bounded-spread unitaries can be used to map between successive codes, such that non-unitary operations are only essential for error correction.

To ensure fault tolerance, code deformation processes must be ELNP. This depends on a number of conditions. First, each quantum channel $\mathcal{U}_i$ should be LNP, which is consistent with Theorem \ref{Th1} because it has bounded spread. Second, each intermediate code family, $\mathcal{C}_i$, must have a threshold for local, stochastic noise, so that there exists a decoder $\mathcal{D}_i$ that gives a threshold on the effective noise models at intermediate times. This second condition is significant and poses a challenge, as these intermediate codes do not necessarily maintain important structures of the original code that ensure a threshold. For example, even if the original code is LDPC, the intermediate codes are not necessarily LPDC nor even stabiliser codes.

An important first example of code deformation is braiding defects in topological (and hypergraph product \cite{Krishna}) codes. However, while such braiding logical operator implementations are not unitary, they nonetheless have bounded spread (as shown in Appendix \ref{AppendixC3}) meaning that they are also constrained by Corollary \ref{Cor1}.

An alternative type of code deformation that can allow Corollary~\ref{Cor1} to be circumvented is \emph{code conversion} \cite{Hill}. Examples can be found by re-expressing conventional code switching~\cite{Paetznik,Anderson,Bombin} as a gradational process. It involves deforming between an initial and final code that together admit a universal set of bounded-spread logical operator implementations, by passing through a sequence of codes such that each mapping between successive codes has bounded spread. One such example is concatenated Steane and Reed-Muller codes \cite{Steane,Steane2}, which are connected by codes that are mixed concatenations of each type of code and so can be mapped between by code conversion (as detailed in Appendix \ref{AppendixD2}) \cite{Hwang,Colladay}. This allows for the (transversal) Clifford group on the concatenated Steane code to be supplemented by a gradational implementation of the non-Clifford $\bar{T}$ gate to realise a universal gate set. 

A second example is \emph{dimensional conversion}, which is inspired by dimensional jumping~\cite{Bombin5}. It is code conversion  between topological codes of different spatial dimensions (discussed in Appendix \ref{AppendixD2}). One such example is between the 3D and 2D surface codes, which are connected by a sequence of 3D surface codes on rectangular prisms. Dimensional conversion can be used to deform between these codes to allow for the non-Clifford  $\overline{\text{CCZ}}$, which completes a universal gate set for the 2D surface code.

Dimensional conversion on higher dimensional codes illuminates the role of non-unitarity in circumventing Corollary~\ref{Cor1}. Specifically, (as discussed in Appendix~\ref{AppendixD2}) we can achieve a universal set of fault-tolerant logical operator implementations by using dimensional conversion between 4D and 6D surface codes such that the codes at all timesteps are self-correcting~\cite{Dennis,Brown4}. Such self-correcting memories allow for passive error-correction, i.e.~dissipation of energy into the environment, to be used instead of active error-correction. This shows that non-unitary dissipative channels can be used to circumvent Corollary~\ref{Cor1}, disproving the common belief that measurements and classical processing are essential to fault-tolerantly implementing a universal gate set.

\subsubsection{Pieceable Fault Tolerance}\label{Sec:IVA2}
Even without a convenient set of intermediate codes, a gradational logical operator implementation can still be possible. An important example is \emph{pieceable fault tolerance} \cite{Knill2,Yoder2}. This approach can be used in the case of logical operators on Calderbank-Shor-Steane (CSS) codes where all the unitary operators, $U_i$, used in the implementation of $\bar{U}$ are diagonal in the computational basis. In this case, $Z$-type stabilisers are left unchanged throughout and so $X$ errors can be corrected in the standard way by measuring $Z$ stabilisers at appropriate intervals and using the initial decoder for such errors. At the same time, $Z$ errors are not spread by such an implementation and so their correction can be deferred until the end of the process when the state has been returned to the original logical space. Such deferral allows $Z$ errors to accumulate uncorrected over the time taken by the logical operator implementation. However, for concatenated codes, this time corresponds to a constant factor per layer of concatenation, and so fault tolerance can be ensured \cite{Knill2,Yoder2}.

For topological codes, simply omitting correction of $Z$ errors at intermediate times would undermine fault tolerance \cite{Yoder2}. This is because the time taken by pieceably fault-tolerant logical operator implementations on a topological code  scales polynomially in the code distance, meaning unacceptably many $Z$ errors can accumulate.
This can be addressed by extracting syndrome information on these accumulated errors by making appropriate $X$ measurements on intermediate states of the code, and later performing the corrections of $Z$ errors based on this information at the end of the implementation of the logical operator. This addition of \emph{delayed error correction} to pieceable fault tolerance can allow for a fault-tolerant gradational implementation on topological codes without fully implementing thresholded decoders on each intermediate code. Indeed,  as  discussed  in  Appendix  \ref{AppendixD3}, this  technique can be understood as the basis for the fault-tolerant implementation of the non-Clifford $\overline{\text{CCZ}}$ logical operator on the two-dimensional surface code in Ref.~\cite{Brown}.

\subsection{Intrinsically Fault-Tolerant Logical Operator Implementations}\label{Sec:IVB}
Intrinsically fault-tolerant logical operator implementations are an alternative approach to circumventing Corollary~\ref{Cor1}.  Unlike gradational implementations, intrinsically fault-tolerant implementations can be performed in constant time (independent of code size). They require that error-correction and the logical operator implementation are performed as a single process, to ensure that all correctable errors present at the time the logical operator is implemented are corrected. In particular, since errors can arise at the beginning of a timestep at which a logical operator is implemented, we cannot simply apply standard error-correction before or after the implementation of the logic gate.

However, there exists an alternative approach to error correction that can meet this requirement:  error-correction can be implemented via teleportation~\cite{Knill3,Knill} (as in Fig.~\ref{Fig:Tel} with $\bar{U}=\bar{I}$).  Furthermore, this technique allows for a logical operator $\bar{U}$ to be implemented simultaneously with the error-correction gadget~\cite{Gottesman} via gate teleportation, as shown in Fig.~\ref{Fig:Tel}. This provides a perspective on why this teleportation-based approach is so effective for implementing fault-tolerant logical operators -- it allows for fault-tolerant implementation of logical operators of unbounded spread by correcting correctable errors while transforming logical operators. The elements of the teleportation gadget that are non-unitary are the adaptive quantum channels -- implementations of logical operators conditional on logical measurement outcomes. These channels propagate logical operators but not correctable errors, as needed to circumvent Corollary~\ref{Cor1}, since the former affect the logical measurement outcome but the latter do not.

\begin{figure}
\centering
\includegraphics[scale=0.87]{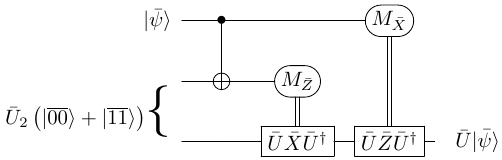}
\caption{Teleportation gadget for intrinsically fault-tolerant implemention of single-qubit logical operator $\bar{U}$. In the case where $\bar{U}=\bar{I}$ it implements error-correction by teleportation. For general $\bar{U}$, this error-correction ensures that $\bar{U}$ is implemented fault-tolerantly. \label{Fig:Tel}}
\end{figure}

In the special case of CSS codes, and logical operator implementations that spread errors only of one type (i.e.,~only $X$ or only $Z$ errors), a simplification of this approach is possible. Specifically, in this case teleportation error-correction is only necessary for the type of error that is spread by the implementation of $\bar{U}$ as the other type of error can be corrected in the standard way at a later time. This gives simplified teleportation gadgets, as shown in Fig.~\ref{Fig:Tel1}.

\begin{figure}
\centering
\subfigure[\label{Fig:Tel1b}]{
\includegraphics{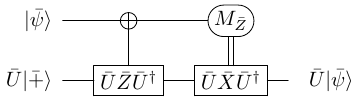}
}
\hspace{3cm}
\subfigure[\label{Fig:Tel1a}]{
\includegraphics{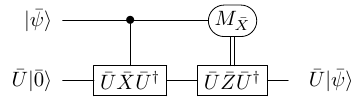}
}
\caption{Simplified teleportation gadgets implementing $\bar{U}$ on CSS codes that are intrinsically fault-tolerant with respect to (a)$X$ errors and (b) $Z$ errors. (The first gate in (a) represents $\overline{\text{CNOT}}_{2,1}$ conjugated by $\bar{U}$ on the control qubit.) \label{Fig:Tel1}}
\end{figure}

Since the teleportation gadget emerges naturally from the requirements of an intrinsically fault-tolerant logical operator implementation, it manifests in a range of apparently distinct approaches to implementing a universal set of fault-tolerant logical operators. In particular, we now briefly review its use in the two most standard approaches: state injection and code switching.

\subsubsection{State Injection}\label{Sec:IVB3}
State injection is a straightforward application of teleportation, named to put greater emphasis on the preparation of the state in the case that the other circuit components admit simple fault-tolerant implementations.

\paragraph{Symmetric Codes -- Magic State Injection:}
Infinite-disjointness codes that are symmetric (i.e.,~those with approximately equal distances of logical Pauli operators) can admit only Clifford bounded-spread logical operators (as shown by the proof of Lemma \ref{Lem4}). Completing a universal gate set thus requires a non-Clifford logical operator, typically $\bar{T}$ (a logical  $\bar{Z}$ rotation by $\pi/4$). Using teleportation to implement such an operator is referred to as \emph{magic state injection}.  Using the circuit in Fig.~\ref{Fig:Tel1b} with $\bar{U}=\bar{T}$, we obtain the circuit for magic state injection shown in Fig.~\ref{Fig:MSI}. The offline preparation of the magic state, $\bar{T}|\bar{+}\rangle$, is generally the most challenging part of the procedure, but this can be achieved (albeit with significant resource cost) by magic state distillation \cite{Bravyi2}. Examples of codes for which magic state injection can be effective include two-dimensional topological codes \cite{Litinski,Chamberland3} and concatenated Steane codes \cite{Chamberland2}.

\begin{figure}
\centering
\centerline{
\includegraphics{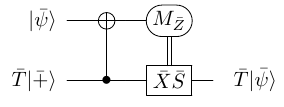}
}
\caption{Magic state injection circuit obtained by letting $\bar{U}=\bar{T}$ in Fig.~\ref{Fig:Tel1b} \label{Fig:MSI}}
\end{figure}

\paragraph{Asymmetric Codes -- Coherent State Injection:}
Asymmetric codes (i.e.~those with significantly different distances of logical Pauli operators) cannot admit bounded-spread logical operator implementations that interchange arbitary Pauli operators (see proof of Lemma \ref{Lem3}). For example, such codes with a single logical qubit do not admit a bounded-spread logical Hadamard operator, which is often sufficient to complete a universal gate set. Letting $\bar{U}=\bar{H}$ in the circuit in Fig.~\ref{Fig:Tel1a} we obtain the circuit shown in Fig.~\ref{Fig:CSI}. We refer to this technique as \emph{coherent state injection} since the injected state  $|\bar{+}\rangle$ is coherent -- meaning that it is a superposition of computational basis states \cite{Takagi}-- and it is this resource that cannot be created by bounded-spread logical operators in asymmetric codes, analogously to magic in symmetric codes. Unlike magic states, coherent states can be stabiliser states and thus can be prepared even in asymmetric codes in a straightforward way. Indeed, they can be prepared in constant-time in codes that admit single-shot error-correction of $X$ errors. Examples of codes for which coherent state injection can complete a universal logical gate set include three-dimensional surface codes \cite{Vasmer,WebsterBraid}, three-dimensional colour codes \cite{Bombin3} and concatenated Reed-Muller codes. We note that coherent state injection can also be used effectively for subsystem codes, such as the Bacon-Shor code~\cite{Yoder}.

\begin{figure}
\centering
\centerline{
\includegraphics{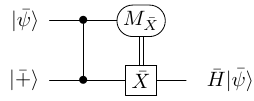}
}
\caption{Coherent state injection circuit obtained by letting $\bar{U}=\bar{H}$ in Fig.~\ref{Fig:Tel1a} \label{Fig:CSI}}
\end{figure}

\subsubsection{Code Switching}\label{Sec:IVB4}
Code switching is an approach by which logical information is teleported between two codes to complete a universal gate set. It can be understood as the intrinsically fault-tolerant counterpart to code conversion. Specifically, we take two codes, $\mathcal{C}_1$ and $\mathcal{C}_2$ and construct fault-tolerant implementations of the logical identity between the codes; $\bar{I}_{1\to 2}: \mathcal{C}_1 \to \mathcal{C}_2$  and its inverse, $\bar{I}_{2\to 1}:\mathcal{C}_2 \to \mathcal{C}_1$. These operators allow any fault-tolerant logical operator, $\bar{U}_2$, of $\mathcal{C}_2$ to be implemented fault-tolerantly on $\mathcal{C}_1$ by $\bar{U}_1=\bar{I}_{2\to 1}\bar{U}_2\bar{I}_{1\to 2}$. They are thus sufficient to lift a set of bounded-spread logical operator implementations to universality in the case that the union of bounded-spread logical operator implementations of $\mathcal{C}_1$ and $\mathcal{C}_2$ is universal. In this case, Theorem \ref{Th2} implies that if $\mathcal{C}_1$ and $\mathcal{C}_2$ together admit bounded-spread implementations of a universal set of logical operators, then $\bar{I}_{1\to 2}$ or its inverse must have unbounded spread.

The essential question for realising a universal gate set by code switching is thus how $\bar{I}_{1\to 2}$ and $\bar{I}_{2\to 1}$ can be implemented. This can be done by applying teleportation with the special case that $\bar{U}=\bar{I}$. Unlike state injection, the primary challenge in implementing this circuit is generally not state preparation, but instead the implementation of the entangling gate between the two codes.

Code switching has been effectively used to circumvent Corollary \ref{Cor1} in topological and concatenated code families. Specifically, dimensional jumping \cite{Bombin5} can be understood as a form of code switching between two-dimensional and three-dimensional colour codes. Similarly, code switching can also be used to switch between concatenated Steane and Reed-Muller codes \cite{Anderson,Paetznik}. While these code switching schemes are typically described in the language of gauge fixing, they can be translated into fundamentally equivalent circuits based on teleportation between different codes, as discussed in Appendix \ref{AppendixD1}.

\section{Conclusion}

In this paper, we have provided a new lens through which to view architecture design for universal, fault-tolerant quantum computing.  This lens takes the form of a general no-go result (Corollary \ref{Cor1}), which excludes the possibility of { implementing} a universal set of logical operators using only LNP quantum channels on a large class of stabiliser (and subsystem) code families, including important examples like concatenated codes and conventional topological stabiliser codes. With this new perspective, we see how natural approaches to universal, fault-tolerant logical operator implementations can be constructed by using non-unitary quantum channels that are not LNP but are nonetheless fault-tolerant, and how a range of apparently disparate schemes that achieve universal fault-tolerant logic can be understood as instances of our approach.

A natural question that remains open is: in what other ways can our no-go result can be circumvented? A natural approach is to consider other classes of fault-tolerant quantum channels that are not LNP. In particular, we could construct quantum channels that map local stochastic noise models to noise models that are not local but nonetheless retain desirable properties that ensure a threshold. Specifically, such channels can allow for local errors to spread to particular errors that have a large support but are nonetheless correctable for appropriately constructed decoders \cite{JochymOConnor2,JochymOConnor4,Wan}. In particular, Jochym O'Connor and Laflamme have presented a universal scheme using non-ELNP logical operator implementations that are constructed in such a way that local errors are mapped to errors that are correctable for any decoder that respects the concatenated structure of the code family \cite{JochymOConnor2}, and Jochym O'Connor has presented a similar scheme based on homological  product codes \cite{JochymOConnor4}. We discuss this further in Appendix \ref{AppendixE}.

Another approach could be to use a code family that is not $\mathbf{B}$-constrained. Such a code family would be unconstrained by Theorem \ref{Th2} and so could conceivably admit a universal set of logical operators implemented by LNP quantum channels. As we have shown, standard code families such as conventional topological stabiliser codes and concatenated stabiliser codes are $\mathbf{B}$-constrained as they necessarily have infinite disjointness. However, it is unknown whether other classes of code families, such as topological stabiliser codes with fracton order, or more general LDPC codes, necessarily have infinite disjointness. In particular, since the disjointness of a code is upper-bounded by the ratio of the number of physical qubits ($n$) to the code distance ($d$)~\cite{JochymOConnor}, if LDPC codes with $d$ linear in $n$ could be found then they could be unconstrained by Corollary~\ref{Cor1}. Such a code family could potentially allow for universal, fault-tolerant quantum computing using only unitary quantum channels and thus may avoid the challenges and high overheads associated with universality in standard code families.

Conversely, if more general LDPC codes are shown to have infinite disjointness and hence are constrained by our no-go results, then the results of Sec.~\ref{Sec:IV} could be applied to them to overcome this obstacle. Indeed, the analysis of that section could serve as a template for constructing a universal set of fault-tolerant logical operations for important code families, such as LDPC codes with asymptotically non-zero rate. This could open up promising new avenues for more efficient and practical fault-tolerant quantum computing.

\begin{acknowledgments}
This work is supported by the Australian Research Council via the Centre of Excellence in Engineered Quantum Systems (EQUS) project number CE170100009. Research at Perimeter Institute is supported in part by the Government of Canada through the Department of Innovation, Science and Economic Development Canada and by the Province of Ontario through the Ministry of Colleges and Universities. T.R.S.~acknowledges support from University College London and the Engineering and Physical Sciences Research Council [grant number EP/L015242/1]. The authors would like to thank Dan Browne, Simon Burton, Robert Harris, Aleksander Kubica and Armanda Quintavalle for helpful discussions.
\end{acknowledgments}

\appendix
\section{Examples of Bounded-Spread Logical Operator Implementations}\label{AppendixC}
In this appendix, we review important classes of bounded-spread logical operator implementations. Specifically, we identify such classes, consider their fault tolerance and show that they indeed have bounded spread. This is important as it shows the breadth of applicability of Theorem \ref{Th2} to a range of types of logical operator implementations.

\subsection{Transversal Gates}
A transversal gate is a logical operator implementation that can be decomposed into a tensor product of physical unitary operators such that no more than one physical qubit from each code block is in the support of each of these operators \cite{Eastin}. Transversal gates are generally assumed to be fault-tolerant on any code family. This assumption is justified for local stochastic noise, since transversal gates are LNP. Indeed, for a local stochastic noise model $\mathcal{E}$ with noise rate $\eta$, and a transversal gate acting on $m$ code blocks implemented by $\mathcal{A}$, $p_{\mathcal{A}}\left(E|\mathbf{A}\subseteq \text{supp}(E)\right)\leq \eta^{\frac{|\mathbf{A}|}{m}}=\left(\eta^{\frac{1}{m}}\right)^{|\mathbf{A}|}$. Since the number of code blocks, $m$, is constant in $l$, this implies that $\mathcal{E}_{\mathcal{A}}$ is a local stochastic noise model with noise rate $\eta_{\mathcal{A}}=\eta^{\frac{1}{m}}$.

Theorem \ref{Th1} thus implies that transversal gates must have bounded spread. Indeed, transversal gates clearly have a spread upper-bounded by the number of code blocks on which they act, which means that they have bounded spread (when acting on a cosntant number of code blocks). This corroborates that transversal gates cannot be universal on $\mathbf{B}$-constrained code families, by Theorem \ref{Th2}. Indeed, the fact that transversal gates do not spread errors at all within a code block means that the requirement of infinite disjointness can be dropped and so Theorem \ref{Th2} can be generalised for transversal gates on any stabiliser code family \cite{JochymOConnor}. In fact, it is known by other means that transversal gates cannot be universal on any code family \cite{Eastin}.

\subsection{Locality-Preserving Logical Operators}
A locality-preserving logical operator is a type of logical operator implementation on topological codes that preserves the geometric locality of errors. Locality-preserving logical operator implementations are typically unitary, but we do not assume that they necessarily are. We do assume, however, that they take constant time (since they only require interactions across regions of constant size) and so such an implementation $\mathcal{A}$ has noisy counterpart $\mathcal{A}_{\mathcal{E}_\eta}=\mathcal{A}\circ \mathcal{E}_\eta$. 

Locality-preserving logical operators also have bounded spread. Indeed, a locality-preserving logical operator, $\mathcal{A}$, maps an operator, $E$, to $\mathcal{A}(E)$ with support contained in an $\epsilon$-neigbourhood around $\text{supp}(E)$, for some constant $\epsilon$. In a $D$-dimensional code, such an $\epsilon$-neighbourhood is contained in the union of $D$-cubes with side length $2\epsilon+1$ centred on each qubit in the support of $E$. Thus, the spread of any error, $E$, under the action of a locality-preserving logical operator is bounded by
\begin{equation}\label{EqLPLOs}
s_{\mathcal{A}}(E)\leq \frac{|\text{supp}(E)|(2\epsilon+1)^D}{|\text{supp}(E)|} = (2\epsilon+1)^D
\end{equation}
and so they indeed have bounded spread. Since they have bounded spread, locality-preserving logical operators are not sufficient for universal quantum computing on infinite-disjointness codes, by Theorem \ref{Th2}. In particular, this applies to conventional topological stabiliser codes, as shown in Appendix \ref{AppendixB}. This is consistent with existing no-go results constraining the gate sets implementable by locality-preserving logical operators \cite{Bravyi,Pastawski,WebsterLPLO}.

Locality-preserving logical operators are also LNP. Indeed, assume that $\mathcal{A}$ is a locality-preserving logical operator and $\mathbf{A}\subseteq \text{supp}\left(\mathcal{A}(E)\right)$ for some set of qubits, $\mathbf{A}$. Then, necessarily, $\mathbf{N}_{\mathbf{A}}\subseteq \text{supp}(E)$, where $\mathbf{N}_{\mathbf{A}}$ is an $\epsilon$-neighbourhood around $\mathbf{A}$. Thus, the probability, that $\mathbf{A}\subseteq \text{supp}\left(\mathcal{A}(E)\right)$ satisfies 
\begin{equation}
p_{\mathcal{A}}(E|\mathbf{A}\subseteq \text{supp}(E))\leq p(\mathbf{A}\cup N_{\mathbf{A}}\subseteq \text{supp}(E))\leq\eta^{|\mathbf{N}_{\mathbf{A}}|}
\end{equation}
As discussed in the previous paragraph, the $\epsilon$-neighbourhood around $\mathbf{A}$ satisfes $|\mathbf{N}_{\mathbf{A}}|\leq (2\epsilon+1)^D|\mathbf{A}|$. Thus, if $\mathcal{E}$ is a local stochastic noise model, then the probability that $\mathbf{A}$ is contained in an error under the effective noise model $\mathcal{E}_{\mathcal{A}}$ satisfies
\begin{equation}
p_{\mathcal{A}}(E|\mathbf{A}\subseteq \text{supp}(E))\leq \eta^{(2\epsilon+1)^D|\mathbf{A}|}=\left(\eta^{(2\epsilon+1)^D}\right)^{|\mathbf{A}|}
\end{equation}
Thus, $\mathcal{E}_{\mathcal{A}}$ is a local stochastic noise model with noise rate $\eta_{\mathcal{A}}=\eta^{(2\epsilon+1)^D}$ and hence $\mathcal{A}$ is LNP. 

\subsection{Braiding Defects}\label{AppendixC3}
Another common approach to fault tolerance in topological codes is the introduction of defects which can be braided to implement logical operators \cite{WebsterBraid,Scruby,Fowler2,Brown3,Bombin6,Kesselring,Krishna2}. More recently, it has been shown that this approach can also be taken in non-topological LDPC codes, such as hypergraph product codes \cite{Krishna}. Braiding defects is a form of code deformation and so can be fault-tolerant provided that decoders that give thresholds can be applied at intermediate times to meet the requirements discussed in Sec.~\ref{Sec:IVA1}. This is relatively easy to achieve in this case, since the fact that each of the intermediate codes differ only by the configuration of defects means the same approach to decoding can be taken in all cases. 

General code deformations can have unbounded spread while being fault-tolerant and so can allow universality. However, it is known that braiding defects in topological stabiliser codes cannot be universal \cite{WebsterBraid}. This is because logical operators implemented by braiding defects have bounded spread and so are constrained by Theorem \ref{Th2}. 

Indeed,  assume that the standard encoding into defects is used \cite{WebsterBraid}. Assume also that all topologically non-trivial paths in a topological stabiliser code with defects all have lengths that are $\Theta(l)$, meaning that they differ by at most a constant factor in the size of the code. This is reasonable for the purpose of understanding the logical operators implementable by braiding, since it can be satisfied by choosing appropriate sizes and separations of defects and the action of a braiding logical operator on the logical space is independent of these factors. Under this assumption, all non-trivial braiding processes are topologically equivalent to a product of processes in which each point on a defect is moved through a distance $\Theta(l)$. Since the image of a logical Pauli operator under the action of a braiding logical operator can be represented by the topological excitation used to implement the original logical operator following a (possibly different) braiding process \cite{WebsterBraid}, and all braiding processes differ by a factor, $c$, bounded in $l$, the spread of any logical Pauli operator, $\bar{P}$, by braiding logical operator implementation $\mathcal{B}$ is bounded by
\begin{equation}
S_{\mathcal{B}}(\bar{P})\leq c
\end{equation}
An example of the bounded spread of logical operators under braiding is shown in Fig.~\ref{FigBraid}. 

Meanwhile, decoding during braiding corrects all local errors (which hence are removed before they are spread beyond a constant-size region), and maps only errors of support size proportional to the code distance to logical Pauli operators, which increases their support size only by a constant. Thus, the spread of errors under braiding is also bounded. Hence, braiding defects necessarily has bounded spread.

Thus, a universal set of logical operator implementations cannot be realised by braiding defects in topological stabiliser codes, in spite of the fact that such implementations are non-unitary.

\begin{figure}
\includegraphics{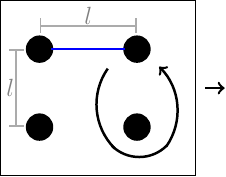}
\includegraphics{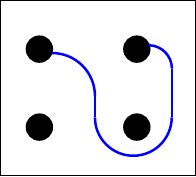}
\caption{Example of a logical operator implemented by braiding defects. The defect configuration is four holes, such that the closest holes are separated by a distance $l$. Braiding one hole around another grows the support of a string-like logical operator of length $l$ (blue) but it remains string-like and of length $\Theta(l)$, so its spread under the braid is constant. \label{FigBraid}}
\end{figure}

\section{Sufficient Conditions for a Code Family to be $\mathbf{B}$-constrained}\label{AppendixSecIIIB}
In this appendix we show that stabiliser code families that are either asymmetric or both symmetric and infinite-disjointness are necessarily $\mathbf{B}$-constrained, as described in Sec.~\ref{Sec:IIIB}. For these proofs, recall that 
\begin{align}
\mathbf{L}_\downarrow &= \left\{\bar{L}\in\mathbf{L}\middle|\exists C\in\mathbb{R}: \frac{d_{\bar{L}}}{d}\leq C \text{ for all }l\in\mathbb{N}\right\}\\
\mathbf{P}_\downarrow &=\mathbf{L}_\downarrow\cap\mathbf{P}_n\\
\mathbf{B} &=\left\langle\left\{\bar{B}\in\mathbf{L}\middle| \bar{B}\bar{P}\bar{B}^\dag \in \mathbf{L}_\downarrow \,\, \forall \bar{P}\in\mathbf{P}_\downarrow \right\}\right\rangle
\end{align}
and we define a stabiliser code family to be $\mathbf{B}$-constrained if $\mathbf{B}$ is not universal. 

\begin{lemma}\label{Lem3}
An asymmetric stabiliser code family is necessarily $\mathbf{B}$-constrained.
\end{lemma}
\begin{proof}
Any logical operator, $\bar{V}$, that is not expressible as a linear combination of elements of $\mathbf{P}_\downarrow$ must instead have expression in the logical Pauli basis that contains a term, $\bar{Q}$, outside $\mathbf{P}_\downarrow$. This implies that for any representative, $V$, of $\bar{V}$, there exists a representative, $Q$, of $\bar{Q}$ such that $\text{supp}(Q) \subseteq \text{supp}({V})$. Hence $d_{\bar{V}}\geq d_{\bar{Q}}$ and so $\bar{V}\not\in \mathbf{L}_\downarrow$. Thus, any element of $\mathbf{L}_\downarrow$ is expressible as a linear combination of elements of $\mathbf{P}_\downarrow$. Thus, $\mathbf{B}$ is a subgroup of the normaliser of the subspace of logical operators spanned by $\mathbf{P}_\downarrow$. Hence the subspace stabilised by $\mathbf{P}_\downarrow$ is invariant under the action of operators in $\mathbf{B}$. Since $\langle\mathbf{P}_{\downarrow}\rangle$ is a non-empty, proper subgroup of $\mathbf{P}_n$, this subspace is a non-trivial, proper subspace. Thus, $\mathbf{B}$ is not universal and so the code family is $\mathbf{B}$-constrained.
\end{proof}

\begin{lemma}\label{Lem4}
An infinite-disjointness, symmetric stabiliser code family is necessarily $\mathbf{B}$-constrained.
\end{lemma}
\begin{proof}
Let $\bar{U}$ be a non-Pauli logical operator. Then, there exists logical Pauli $\bar{P}$ such that $[\bar{U},\bar{P}]$ is not a scalar multiple of $\bar{I}$, which implies that for the smallest-support representative $P$ of $\bar{P}$ and $U$ that is the smallest-support representative of $\bar{U}$
\begin{equation}\label{Lem4Eq1}
|\text{supp}(P)\cap \text{supp}(U)|\geq\left|\text{supp}\left([U,P]\right)\right|\geq d_{\bar{V}}
\end{equation}
for some logical operator $\bar{V}$.

By the scrubbing lemma \cite{JochymOConnor},
\begin{equation}
\Delta|\text{supp}(P)\cap \text{supp}(U)|\leq |\text{supp}(U)|
\end{equation}
Rearranging and using that $d_{\bar{U}}=|\text{supp}(U)|$ and Eq.~(\ref{Lem4Eq1}), this implies
\begin{equation}\label{Lem4Eq3}
\frac{d_{\bar{U}}}{d_{\bar{V}}} \geq \Delta
\end{equation}
Since the code family is assumed to be infinite-disjointness, then taking the limit of Eq.~(\ref{Lem4Eq3}) implies
\begin{equation}
\lim_{l\to\infty} \frac{d_{\bar{U}}}{d_{\bar{V}}} =\infty
\end{equation}
Thus, $\bar{U}\not\in \mathbf{L}_\downarrow$. 

Hence, $\mathbf{L}_\downarrow$ is a subset of the logical Pauli group and so also there is some non-trivial logical Pauli $\bar{P}$ in $\mathbf{L}_\downarrow$. Since the code family is assumed to be symmetric, there is some generating set $\mathbf{P}_g$ of the logical Pauli group such that $\lim_{l\to\infty} \frac{d_{\bar{Q}}}{d_{\bar{P}}}\neq\infty$ for any $\bar{Q}\in\mathbf{P}_g$. Thus, $\mathbf{P}_g\subseteq \mathbf{L}_\downarrow$ and so all elements of $\left\{\bar{B}\in\mathbf{L}\middle| \bar{B}\bar{P}\bar{B}^\dag \in \mathbf{L}_\downarrow \,\, \forall \bar{P}\in\mathbf{P}_\downarrow \right\}$ are logical Clifford operators. Thus, $\mathbf{B}$ is contained in the logical Clifford group, which is not universal.
\end{proof}

\section{Disjointness of Code Families}\label{AppendixB}
The disjointness is a measure of how many mostly disjoint representatives of logical Pauli operators a stabiliser code has \cite{JochymOConnor}. Specifically, the $1$-disjointness is the largest integer $\Delta_1$ such that each logical Pauli operator admits at least $\Delta_1$ representatives with mutually disjoint support (i.e.~such that no qubit is in the support of more than one of the representatives for each logical operator). The disjointness of a code is lower-bounded by its $1$-disjointness. In particular, for a code to be infinite-disjointness it is sufficient that $\lim_{l\to\infty} \Delta_1(l)=\infty$.

The disjointness is a generalisation of the $1$-disjointness that also accounts for almost disjoint representatives. Specifically, the unnormalised $c$-disjointness (for $c\geq 1$) is the largest integer, $\Delta_c'$, such that each logical Pauli operator admits at least $\Delta_c'$ representatives that are $c$-disjoint, i.e.~such that no qubit is in the support of more than $c$ of the representatives for each logical operator. The (normalised) $c$-disjointness is the unnormalised $c$-disjointness divided by $c$: $\Delta_c=c^{-1} \Delta_c'$. The disjointness of the code is the maximum value taken by the $c$-disjointness for any $c$:
\begin{equation}
\Delta=\max_{c\in \mathbb{Z}_+} \Delta_c
\end{equation}

In this appendix, we show that concatenated stabiliser code families and conventional topological stabiliser codes families necessarily have infinite disjointness. In each case, we also comment on possible extensions.

\subsection{Concatenated Codes}\label{AppendixB1}
Concatenation of quantum-error correcting codes is a natural and effective technique for constructing code families. Specifically, for any $[[n,1,d]]$ quantum error-correcting code, $\mathcal{C}$, there exists a code family such that $\mathcal{C}_l$ is an $[[n^l,1,d^l]]$ code constructed by the $l$-fold concatenation of $\mathcal{C}$. More generally, we may concatenate different codes -- by concatenating a sequence of $[[n_i,1,d_i]]$ codes, we can construct a code family such that $\mathcal{C}_l$ is a $[[\prod_{i=1}^l n_i,1,\prod_{i=1}^l d_i]]$ code constructed by the concatenation of the first $l$ codes in the sequence. (We assume here that $n_i$ is bounded in $i$, since the idea is that the codes used in the concatenation should be small codes that are used as building blocks for the large concatenated code.) The quantum threshold theorem implies that concatenated code families indeed have a threshold against local stochastic noise \cite{Aharanov,Aliferis}.

In order to prove that concatenated stabiliser codes (i.e.~concatenated code families where all codes in the concatenation are stabiliser codes) have infinite-disjointness, we first prove the following lemma showing that the disjointness of a code constructed by concatenating a pair of codes is at least the product of the disjointnesses of the two codes.
\begin{lemma}\label{Lem6}
Let $\Delta$ be the disjointness of a concatenated code constructed by encoding the physical qubits of code 1, with disjointness $\delta_1$ into code 2, with disjointness $\delta_2$. Then, $\Delta \geq \delta_1\delta_2$.
\end{lemma}
\begin{proof}
Logical operator representatives of the concatenated code are constructed by specifying a representative of the logical operator on code $1$, and then, for each non-trivial Pauli operator in this representative, specifying a representative of the corresponding logical operator on code $2$. For each of the codes, $i=1,2$, there exists $c_i$ such that there is a set of $c_i\delta_i$ $c_i$-disjoint representatives of each logical Pauli operator of code $i$. Using combinations of logical operator representatives from these sets, we can construct a set of $c_1c_2\delta_1\delta_2$ representatives of any logical operator on the concatenated code. Each physical qubit is in the support of one of these representatives only if the chosen code $1$ logical operator representative is one of the at most $c_1$ with support on the block that contains this qubit and the chosen code $2$ logical operator representative is one of the at most $c_2$ with support on the corresponding qubit within the block. Thus, each qubit is in the support of at most $c_1c_2$ of the $c_1c_2\delta_1\delta_2$ logical operator representatives. Hence,
\begin{equation}
\Delta \geq \frac{c_1c_2\delta_1\delta_2}{c_1c_2}= \delta_1\delta_2
\end{equation}
\end{proof}

We now prove that concatenated codes are necessarily infinite-disjointness.
\begin{theorem}\label{Th3}
Concatenated stabiliser codes are infinite-disjointness
\end{theorem}
\begin{proof}
The disjointness is a ratio of cardinalities of sets of representatives of logical Pauli operators, which are necesarily upper-bounded by the number of representatives of each logical Pauli operator (which is equal to the size of the stabiliser group, $2^{n-k}$). Thus, since the disjointness of any non-trivial code satisfies $\Delta>1$ \cite{JochymOConnor}, it follows that the disjointness is bounded by $\Delta\geq \frac{2^{n-k}}{2^{n-k}-1}$. Since each code used in the concatenation has a number of physical qubits, $n$, bounded by some constant $N$, this implies that there exists a constant $C=\frac{2^{N-k}}{2^{N-k}-1}>1$ such that the disjointness of each code used in the concatenation is at least $C$.

By induction, the disjointness, $\delta_l$, of the code after $l$ levels of concatenation thus satisfies $\delta_l\geq C^l$. Indeed, this is trivial for $l=1$ and, assuming that it is true for $l=m-1$, it follows by Lemma \ref{Lem6} that 
\begin{equation}
\delta_l \geq \delta_{l-1}\delta\geq C^{l-1}C \geq C^l
\end{equation}
where $\delta$ is the disjointness of the $l$th code in the concatenation.

Thus, the disjointness of the concatenated code, $\Delta(l)=\delta_l$ satisfies
\begin{equation}
\lim_{l\to\infty} \Delta(l) \geq \lim_{l\to\infty} C^l=\infty
\end{equation}
\end{proof}

By Theorem \ref{Th2} and Corollary \ref{Cor1} respectively, we thus have the following no-go results for concatenated codes.
\begin{corollary}
No concatenated stabiliser code family can admit a universal set of bounded-spread logical operator implementations.
\end{corollary}

\begin{corollary}
{ A universal set of logical operators cannot be implemented on any concatenated stabiliser code by local noise preserving quantum channels.}
\end{corollary}

These results can also be generalised to concatenated subsystem codes by applying the same arguments to disjointness of bare logical operators (as discussed in Sec.~\ref{Sec:IIIB}).

We also note that concatenated codes are a special case of tensor network codes \cite{Farrelly} that have a tree structure. This tree structure straightforwardly gives rise to the property of having infinite disjointness -- increasing $l$ corresponds to growing the tree which increases the number of indepenent blocks on which logical operator representatives can act, implying that the number of (almost) disjoint such representatives grows exponentially in $l$. An interesting question for future work concerns the disjointness of other classes of tensor network codes, especially holographic codes \cite{Pastawski2,Harris1,Harris2}. Holographic codes have a similar structure to concatenated codes, but correspond to more complex tensor networks in which a physical qubit can be dependent on multiple tensors closer to the centre of the network. This dependence makes it less clear whether sufficiently many almost disjoint logical operator representatives can necessarily be constructed in holographic codes for them to have infinite disjointness. We note that if they, or other classes of tensor network codes, do not necessarily have infinite disjointness then they could be promising candidates to circumvent our results and could possibly realise a universal set of LNP implementations of logical operators.

\subsection{Conventional Topological Stabiliser Codes}\label{AppendixB2}
Topological codes naturally correspond to code families. Specifically, a topological stabiliser code is defined by a stabiliser group generated by operators that are geometrically local on a lattice in $D\geq 2$ dimensions. This group must be chosen such that all logical operators have support on a non-local region, meaning a region of at least the size of the shortest path across the lattice which we define as $l$.
This defines a code family as logical error rates can be arbitrarily suppressed by increasing $l$.

Within the class of topological stabiliser codes, we can identify a subclass we refer to as \emph{conventional topological stabiliser codes}. This subclass includes all two-dimensional topological stabiliser codes, all scale and translationally symmetric (STS) codes \cite{Yoshida} including the surface and colour codes in all spatial dimensions, and all topological stabiliser codes with defects that use the standard encoding as defined in Ref.~\cite{WebsterBraid}. Conventional topological stabiliser codes are topological stabiliser codes that have conventional topological order \cite{Dua}, meaning that they can be associated with topological quantum field theories with topological excitations that correspond to error syndromes and which can propagate freely through the code. Logical Pauli operators on conventional topological stabiliser codes have representatives that correspond to topologically equivalent paths of an excitation through the code, such as paths around a handle of a torus, around a topological defect, or between two code boundaries.

Conventional topological stabiliser codes necessarily have infinite disjointness, as we show below. This implies that, by Theorem \ref{Th2} and Corollary \ref{Cor1} respectively, we have the following no-go results for conventional topological stabiliser codes.
\begin{corollary}
No conventional topological stabiliser code family can admit a universal set of bounded-spread logical operator implementations.
\end{corollary}

\begin{corollary}
{ A universal set of logical operators cannot be implemented on any conventional topological stabiliser code by local noise preserving quantum channels.}
\end{corollary}

While the proof of infinite disjointness for general conventional topological stabiliser codes that follows is quite technical, we note that it is straightforward to see for a number of the most common instances. Indeed, in Fig.~\ref{Fig:IDTSC} we show visually how some of these instances can easily be seen to have infinite disjointness.

We also note that it is unclear whether more general topological stabiliser codes necessarily have infinite disjointness. The property for conventional topological stabiliser codes emerges from their TQFT structure; increasing $l$ can be understood to be like fine-graining the model to allow for a greater number of topologically equivalent paths of excitations. This generalises to encodings on any topological stabiliser code that depend only on \emph{topological excitations}, as in the defect encodings considered in Ref.~\cite{WebsterBraid}. However, more general topological stabiliser codes (in $D\geq 3$ dimensions) can have more exotic excitation structures integral to their encoding \cite{Dua}, which makes it unclear whether they necessarily must also have infinite disjointness. In particular, it is known that the locality-preserving logical operators admitted by such codes cannot be universal \cite{Bravyi}, but if finite-disjointness topological stabiliser codes are possible it could nonetheless be the case that more general LNP implementations of logical operators could be sufficient for universality. This is an interesting question for future work.

\begin{figure*}
\centering
\subfigure[]{\label{fig:a}
\includegraphics[scale=0.9]{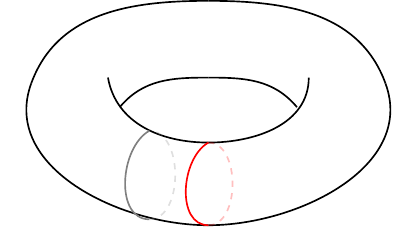}}
\subfigure[]
{\label{fig:b}
\includegraphics{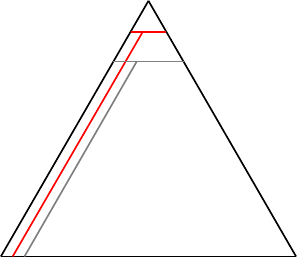}
}
\subfigure[\label{3DSC}]{
\includegraphics{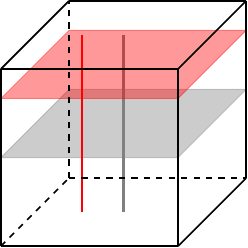}
}
\caption{Illustrations of infinite disjointness of examples of conventional topological stabiliser codes. \newline(a) Toric code: Logical operator representatives correspond to topologically non-trivial loops around the torus. Specifically, $\bar{X}_1$ and $\bar{Z}_2$ are vertically oriented loops, $\bar{X}_2$ and $\bar{Z}_1$ are horizontally oriented loops and products such as $\bar{Y}_1$ and $\bar{Y}_2$ are diagonally oriented loops (which go around both holes). Disjoint representatives of each logical operator correspond to parallel loops around the torus (e.g. translating the red loop gives a disjoint, topologically equivalent grey loop). The disjointness of the code is thus lower-bounded by the circumference around the torus, which is grows with $l$. \newline (b) Colour code: Logical operator representatives correspond to topological excitations that traverse to all three boundaries. Labelling the height of the triangle, $l$, we can construct $l$ representatives (of the kind shown in red and grey) such that each qubit is in the support of at most two of the representatives). The disjointness is thus bounded by $\frac{l}{2}$ which goes to infinity \newline(c) 3D surface code: $\bar{Y}$ logical operator representatives correspond to the paths of point-like charge between top and bottom faces and a line-like flux stretching from front to back passing from left face to right. We can construct $l$ representatives (of the kind shown) so that each qubit is only in the support of at most two of the representatives. Since the $\bar{X}$ and $\bar{Z}$ operators are simply planes and columns and so have $l$ disjoint representatives, this implies that the disjointness of the code is at least $\frac{l}{2}$. \label{Fig:IDTSC}}
\end{figure*}

\subsubsection{Proof of Infinite Disjointness}
A code on a given lattice has a set of paths through the lattice. We allow such paths to have discontinuities, provided they are locally continuous. We define two paths to be topologically equivalent if one can be deformed into the other. We thereby define equivalence classes of paths, $\alpha_i$, such that two paths are in the same class iff they are topologically equivalent. Such classes of paths can be defined by, for example, the boundaries they pass between or the defects (such as holes or twists) that they enclose. We define a notion of product of classes whereby $\alpha_{ij}=\alpha_i\alpha_j$ denotes the equivalence class of paths that are a representative path of $\alpha_i$ followed by a representative path of $\alpha_j$.  We can specify an independent generating set of classes of paths, $A=\{\alpha_j\}$ with respect to this product.

A logical Pauli operator, $\bar{P}$ on a conventional topological stabiliser code can be defined by associating a list of excitations of the TQFT, $E_{\bar{P}}=\{p_j\}$, with each class of paths $\alpha_j\in A$. By this, we mean that that representatives of $\bar{P}$ are compositions of $p_j$ traversing a representative path of $\alpha_j$, for all $j$. We refer to a generating path class $\alpha_j\in A$ as \emph{encoding} if there is some $\bar{P}$ for which $p_j$ is not trivial (i.e.~such that an excitation that is not the vacuum must traverse a path in $\alpha_j$ to implement $\bar{P}$). The following lemma shows that encoding generating path classes necessarily have as many disjoint representatives as the code distance (which is necessarily extensive in $l$).
\begin{lemma}\label{LemTSC}
Let $\alpha_j\in A$ be an encoding generating path class on a code with distance $d$. Then, there exist at least $d$ disjoint representatives of $\alpha_j$.
\end{lemma}
\begin{proof}
Assume that there are not $d$ disjoint representatives. Then, there exists some region of the code that is a \emph{bottleneck} for $\alpha_j$, by which we mean that there is a set of less than $d$ qubits on which all representatives must have support. However, by the cleaning lemma \cite{Bravyi5}, we can necessarily clean all logical operators off this bottleneck so that it acts trivially there. Thus, no logical operator can have a non-trivial excitation on $\alpha_j$, and so $\alpha_j$ is not an encoding path.
\end{proof}

We now prove that conventional topological stabiliser codes necessarily have infinite disjointness. To do so, we make the assumption that the number of independent generating path classes is constant in $l$, which is true if the genus of the mainfold and the number of boundaries and defects are all constant. This assumption is justified by the fact that code families are intended to provide improved protection on a given topology by increasing code size, rather than altering this topology.
\begin{theorem}
Conventional topological stabiliser codes are infinite-disjointness.
\end{theorem}
\begin{proof}
Let $\bar{P}$ be a logical Pauli operator. Representatives of $\bar{P}$ correspond to a choice of representative paths from each encoding generating path class, $\alpha_j$. By Lemma \ref{LemTSC}, we can choose $d$ representatives of $\bar{P}$ such that the representatives from each $\alpha_j$ are disjoint. Thus, the number of these representatives of which any qubit is in the support is at most the number of encoding generating path classes which by assumption is a constant in $l$, $C$. Thus, the disjointness is bounded by $\Delta=\frac{d}{C}$ and so $\lim_{l\to\infty} \Delta=\infty$.
\end{proof}

\section{Proof that ELNP Unitary Quantum Channels Cannot Be Universal }\label{AppendixA}
In this section, we present a complete proof of the fact, stated in Sec.~\ref{Sec:IVAA}, that unitary ELNP quantum channels cannot allow circumvention of Corollary \ref{Cor1}. This justifies the focus on non-unitary quantum channels, which are not constrained by this fact.

\begin{lemma}\label{LemA1}
For a unitary quantum channel, $\mathcal{U}$, there exists a single-qubit Pauli error, $P_j$ such that $s_{\mathcal{U}}=s_{\mathcal{U}}(P_j)$.
\end{lemma}
\begin{proof}
By definition of spread, there exists an error, $P\in\mathbf{P}_n$ such that $s_{\mathcal{U}}=\frac{\left|\text{supp}\left(UPU^\dag\right)\right|}{|\text{supp}\left(P\right)|}$. $P$ can be decomposed into a product of single-qubit Pauli errors, $P=\prod_{i=1}^{|\text{supp}(P)|}P_i$. Using this decomposition,
\begin{align}
s_{\mathcal{U}} &=\frac{|\text{supp}\left(UPU^\dag\right)|}{|\text{supp}(P)|} \\
&= \frac{\left|\text{supp}\left(U\left(\prod_{i=1}^{|\text{supp}(P)|}P_i\right)U^\dag\right)\right|}{|\text{supp}(P)|}\\
&=\frac{\left|\text{supp}\left(\prod_{i=1}^{|\text{supp}(P)|} UP_iU^\dag\right)\right|}{|\text{supp}(P)|} \label{EqHom}\\
&\leq \frac{\sum_{i=1}^{|\text{supp}(P)|} |\text{supp}\left(UP_iU^\dag\right)|}{|\text{supp}(P)|}\\
&\leq \frac{|\text{supp}(P)| \max_{1\leq i \leq |\text{supp}(P)|}  |\text{supp}\left(UP_iU^\dag\right)|}{|\text{supp}(P)|}\\
&= \max_{1\leq i \leq |\text{supp}(P)|}  |\text{supp}\left(UP_iU^\dag\right)|\\
&=  \max_{1\leq i \leq |\text{supp}(P)|}  s_{\mathcal{U}}(P_i)
\end{align}
Thus, there exists a single qubit error, $P_j$ such that $s_{\mathcal{U}} \leq s_{\mathcal{U}}(P_j)$. By definition of spread, this implies that $s_{\mathcal{U}} = s_{\mathcal{U}}(P_j)$.
\end{proof}

\begin{theorem}\label{ThA}
A unitary quantum channel that is ELNP necessarily has bounded spread.
\end{theorem}
\begin{proof}
For any single-qubit Pauli error, $P_i$, there exists a local stochastic noise model such that $p(P_i)=\eta$ and $p(I)=1-\eta$, where $\eta\neq 0$ is the noise rate. If that noise model is applied at the first timestep, and the noise is trivial at subsequent timesteps, then the effective noise model, $\mathcal{E}_{\mathcal{U}}$, has a probability of $\mathcal{U}(P_i)$ given by $p_{\mathcal{U}}(\mathcal{U}(P_i))=p(P_i)=\eta$. Thus, $\mathcal{E}_{\mathcal{U}}$ can only be essentially local if $\mathcal{U}(P_i)$ has a support size bounded in the code size, for any single-qubit Pauli operator $P_i$.  By Lemma \ref{LemA1}, this implies that $\mathcal{U}$ has bounded spread.
\end{proof}
The following corollary follows immmediately from Theorem \ref{Th2} and Theorem \ref{ThA}:
\begin{corollary}\label{Cor2}
A universal set of logical operators cannot be implemented by unitary ELNP quantum channels.
\end{corollary}

\section{Examples of Fault-Tolerant, Unbounded-Spread Logical Operator Implementations}\label{AppendixD}
In this appendix, we provide further details on approaches to universal fault-tolerant quantum computing. Specifically, we first consider code switching, both as an intrinsically fault-tolerant process and as a gradational process. We then consider pieceable fault tolerance with delayed error correction.

\subsection{Code Switching}\label{AppendixD1}
Code switching is a useful technique between two codes that have sets of logical operators admitting bounded-spread implementations such that the union is a universal gate set. We have examples of such codes for both concatenated and topological code families. In particular, the Steane code admits transversal implementations of the full Clifford group, while the Reed-Muller code admits a transversal implementation of a non-Clifford ($\bar{T}$) logical operator. (See Table \ref{Table2} for the definitions of these codes.) These codes can be concatenated to construct code families -- the concatenated Steane code and concatenated Reed-Muller code respectively. Alternatively, the codes can be generalised to (conventional) topological stabiliser code families -- the Steane and Reed-Muller codes are the smallest instances of the two-dimensional triangular and three-dimensional tetrahedral colour codes respectively. In both cases, the code families inherit the transversal gates of the original codes. Code switching between the Steane and Reed-Muller codes can thus be used as a basis for realising a universal gate set with these code families.

\begin{table}
\resizebox{0.48\textwidth}{!}{
  \begin{tabular}{| c | c | c |}
    \hline
     & \textbf{Steane Code} & \textbf{Reed-Muller Code} \\ \hline
    \textbf{Stabiliser} & $X_1X_2X_3X_4$  & $X_1X_2X_3X_4X_8X_9X_{10}X_{11}$ \\ 
\textbf{Generators} & $X_2X_3X_5X_6$ & $X_2X_3X_5X_6X_9X_{10}X_{12}X_{13}$ 
\\
& $X_3X_4X_5X_7$ & $X_3X_4X_5X_7X_{10}X_{11}X_{12}X_{14}$  
\\
&  & $X_8X_9X_{10}X_{11}X_{12}X_{13}X_{14}X_{15}$
\\
& $Z_1Z_2Z_3Z_4$ & $Z_1Z_2Z_3Z_4$\\
& $Z_2Z_3Z_5Z_6$ & $Z_2Z_3Z_5Z_6$\\
& $Z_3Z_4Z_5Z_7$ & $Z_3Z_4Z_5Z_7$\\
&  & $Z_8Z_9Z_{10}Z_{11}$ \\
& & $Z_9Z_{10}Z_{12}Z_{13}$ \\
& & $Z_{10}Z_{11}Z_{13}Z_{14}$ \\
& & $Z_{12}Z_{13}Z_{14}Z_{15}$\\
& & $Z_1Z_4Z_8Z_{11}$\\
& & $Z_2Z_5Z_9Z_{12}$\\
& & $Z_6Z_7Z_{13}Z_{14}$\\
\hline
    \textbf{Logical} & $X_1X_2X_3X_4X_5X_6X_7$ & $X_1X_2X_3X_4X_5X_6X_7$ \\
\textbf{Operators} & $Z_1Z_2Z_3Z_4Z_5Z_6Z_7$ & $Z_1Z_2Z_3Z_4Z_5Z_6Z_7$\\
\hline
  \end{tabular}}
\caption{Stabiliser generators and logical operators of the Steane and Reed-Muller codes. \label{Table2}}
\end{table}

This code switching can be done fault-tolerantly by using the circuits in Fig.~\ref{Fig:Tel1}. Specifically, there is a transversal implementation of $\overline{\text{CNOT}}$ with the control on the Reed-Muller (or concatenated Reed-Muller or 3D colour) code and target on the Steane (or concatenated Steane or 2D colour) code. Combined with standard techniques for preparing stabiliser states and performing fault-tolerant logical Pauli measurements in CSS codes, this allows for a logical state of the Reed-Muller (or concatenated Reed-Muller or 3D colour) code to be mapped to the same logical state of the Steane (or concatenated Steane or 2D colour) code by using the circuit in Fig.~\ref{Fig:Tel1a} with the first wire being the former code and the second code the latter. The inverse mapping can similarly be implemented using the circuit in Fig.~\ref{Fig:Tel1b}.

We emphasise that while codes based on Steane and Reed-Muller codes serve as useful examples, code switching can be used more generally to achieve a universal fault-tolerant gate set. Of particular importance for its relationship to other schemes discussed below, we note that the 2D and 3D surface codes similarly admit a tranversal implementation of $\overline{\text{CNOT}}$ between them, and so switching can similarly be performed between them. The 2D surface code admits bounded-spread implementations of the logical Clifford group, while the 3D code admits a bounded-spread implementation of a non-Clifford ($\overline{\text{CCZ}}$) logical operator \cite{KubicaUnfolding,WebsterLPLO,Vasmer}, and so a universal set of fault-tolerant logical operator implementations can be realised by code switching between these codes.

Our notion of code switching is fundamentally similar to a number of important schemes for fault tolerant universal sets of logical operator implementations. { Indeed, we note that it is closely related to that described in Ref.~\cite{Beverland2}. More generally, } Ref.~\cite{Anderson} presents a technique for switching between Steane and Reed-Muller codes by making appropriate stabiliser measurements of the code being switched into, and applying appropriate corrections conditionally on the measurement outcomes. This scheme can be understood to differ from the code switching technique we have described only in that the physical qubits of the Steane code are considered to be a subset of those of the Reed-Muller code. Due to the related structures of the two codes, this allows for operations that act only on these qubits (specificially, the transversal $\overline{\text{CNOT}}$ and measurements on these qubits) to be omitted, producing a potentially more efficient scheme. This technique can be generalised to concatenated Steane and concatenated Reed-Muller codes by performing the same operations at all levels of concatenation. While this scheme is physically different from our code switching technique, we emphasise that it retains the essential spirit. Specifically, the measurements performed during switching and corresponding corrections can be understood to introduce non-unitarity into the logical operator implementation, by allowing for correctable errors to be detected and corrected during the implementation. Since these corrections can be necessary even in the absence of noise, and can only be corrected based on measurements made during the logical operator implementation, this process differs from standard error correction, and thus highlights the implementation's intrinsically fault-tolerant nature. As described in Ref.~\cite{Anderson}, the scheme can also be related to using gauge fixing of a subsystem code to realise the Steane and Reed-Muller codes as required \cite{Paetznik}. This process of gauge fixing similarly is an intrinsically fault-tolerant process that uses measurement and adaptive corrections to circumvent Corollary \ref{Cor1}.

Dimensional jumping, which switches between topological stabiliser codes of different dimensions, can similarly be understood as a variation on our code switching scheme. Indeed, switching between 3D and 2D colour codes in the way described in Ref.~\cite{Bombin5} is similar to the scheme of Ref.~\cite{Anderson}, in that it corresponds to the case where the physical qubits of the 2D code are a subset of those of the 3D code, which again allows some operations to be omitted without fundamentally altering the approach used to circumvent Corollary \ref{Cor1}. Similar to above, this procedure can also be equivalently understood in terms of fixing the gauge of a subsystem code (the gauge colour code) \cite{Bombin}.

\subsection{Code Conversion}\label{AppendixD2}
While code switching can be implemented intrinsically fault-tolerantly, it is also possible to modify it into a gradational process -- code conversion. Code conversion has the advantage that all measurements and adaptively applied operators correspond to standard error-correction on the intermediate codes. This means that, in the case where decoders can be used on the intermediate codes that are analogous to those of the codes being switched between, standard error-correction and unitary implementations of logical operators can be sufficient for switching and hence for a universal gate set. Indeed, in the case where all codes involved are self-correcting, it allows for a universal gate set without requiring any measurements and adaptive gates.

As a first example, note that it is straightforward to construct a code conversion scheme for concatentated codes, such as the concatenated Steane and Reed-Muller codes. To do so, we first construct a unitary circuit to switch between Steane and Reed-Muller codes. With qubits numbered as in Table \ref{Table2}, such a circuit for switching from the Steane code to the Reed-Muller code is $\prod_{i=1,5,7}\text{CNOT}_{i+7,15}\prod_{i=1}^7 \text{CNOT}_{i,i+7}$, where controls are on the Steane code and targets on the Reed-Muller code. The inverse of this circuit, $\prod_{i=1}^7 \text{CNOT}_{i,i+7}\prod_{i=1,5,7}\text{CNOT}_{i+7,15}$, switches from the Reed-Muller code to the Steane code. Switching between the concatenated codes can then be performed by applying this circuit at one level of concatenation at a time. Since the CNOT operators are transversal at lower levels of concatenation, the spread of this circuit applied at a single level of concatenation does not scale with $l$ and so remains bounded. Indeed, such a circuit applied at a single level of concatenation is fault-tolerant. 

It is straightforward to perform error correction at each intermediate step, since the intermediate codes remain concatenated codes and so can be decoded one level at a time. Thus, a code conversion can be performed, $\mathcal{I}=\prod_{i=1}^l \mathcal{D}_i \mathcal{U}_i$, where $l$ is the number of levels of concatenation, each $\mathcal{D}_i$ is a decoder for the appropriate concatenated code and $\mathcal{U}_i$ is the unitary quantum channel performing the necessary code switch at the $i$th level of concatenation. This allows the universal set of logical operator implementations to be achieved, analogously to for intrinsically fault-tolerant code switching, using gradational logical operator implementations. We note that this scheme is related to other proposed code switching schemes \cite{Hwang,Colladay}, but differs in that we use unitary quantum channels instead of measurements to perform switching.

\subsubsection{Dimensional Conversion}
It is also possible to perform code switching between different topological stabiliser codes -- dimensional jumping -- gradationalally. In particular, we consider the example of code conversion between (planar) surface codes in different dimensions, called dimensional conversion.

Specifically, we consider a three-dimensional surface code on an $l\times l\times l$ cubic lattice and a two-dimensional surface code on an $l\times l$ square lattice. These two codes together have a universal set of bounded-spread logical operator implementations (Clifford + $\overline{\text{CCZ}}$) and are connected by a sequence of three-dimensional surface codes on asymmetric $l\times l \times a$ lattices, for $1\leq a< l$ (as shown in Fig.~\ref{Fig:ADJFull}). Each of these intermediate codes can be error-corrected using any valid three-dimensional surface code decoder. Thus, code conversion between these codes can be performed if unitary, bounded-spread implementations for switching between codes on $l\times l \times a$ and $l\times l \times (a+1)$ lattices can be found (for $0\leq a<l$). 

\begin{figure*}
\includegraphics{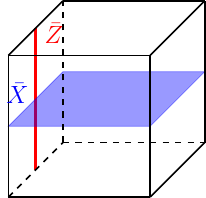}
\includegraphics{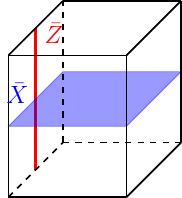}
\includegraphics{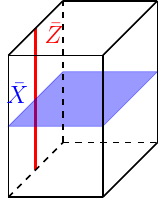}
\includegraphics{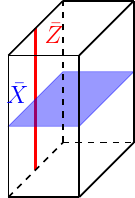}
\includegraphics{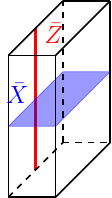}
\includegraphics{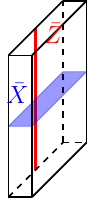}
\includegraphics{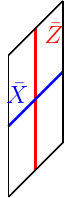}
\caption{Sequence of surface codes through which dimensional conversion deforms to transform between the standard three-dimensional ($l\times l \times l$) surface code (far left) and two-dimensional ($l\times l$) surface code (far right). \label{Fig:ADJFull}} 
\end{figure*}

We can indeed perform such switches. Specifically, a 2D layer can be added to produce an $(l\times l\times (a+1))$ surface code by using a locality-preserving, unitary circuit to appropriately entangle the additional layer to the boundary of an $(l\times l \times a)$ surface code. Indeed, this can be done using a product of nearest-neighbour CNOT operators, as described in Sec.~IX of Ref.~\cite{Higgott}. A layer can be removed by performing the inverse circuit between the boundary that is to be removed and the adjacent 2D layer of the code (see Fig.~\ref{Fig:ADJ}). Since both switches are {locality-preserving}, they necessarily have bounded-spread as required. We note that successively adding layers gives a unitary circuit to encode a 3D surface code from 2D slices, similar to Ref.~\cite{Higgott}, but with the addition of intermediate error-correction to ensure fault tolerance.

The process of dimensional conversion can also be applied between more general pairs of surface codes by an analogous process. In particular, we note that dimensional conversion between the symmetric (dim($\bar{X}$)=dim($\bar{Z}$)=2) 4D surface code and an asymmetric (dim($\bar{X}$)=4, dim($\bar{Z}$)=2) 6D surface code can be achieved and is theoretically interesting. Indeed, it can be performed by adding/removing 4D surface code layers to switch between a 4D and 5D code, combined with adding/removing 5D surface code layers to switch between 5D and 6D. 

The significance of this case is that the union of the bounded-spread logical operator implementations of these codes is universal  (Clifford for 4D + $\overline{\text{CCZ}}$ for 6D) \cite{WebsterLPLO} but both codes, and all intermediate codes, are self-correcting (as they all have logical Pauli operators supported on manifolds of at least two dimensions) \cite{Dennis,Brown4}. A self-correcting code has a logical space that is the ground space of a local Hamiltonian with the property that high weight errors correspond to high energy states and so are naturally suppressed. This means that the active error-correction applied after adding or removing each layer can be replaced by passive errorr-correction. This passive error correction can be performed by dissipation of energy into the environment, without requiring measurements or classical processing for adaptive corrections. Thus, it shows that none of adaptivity, classical processing nor even measurements are necessary to circumvent Corollary \ref{Cor1} -- unitary channels supplemented by non-unitary dissipation is sufficient. This is noteworthy since all other universal schemes we have seen depend on these ingredients in some way, either explicitly in intrinsically fault-tolerant implementations, or implicilty by active error-correction. Hence, this example supports the reasoning in the Sec.~\ref{Sec:IVAA} that it is differential treatment of logical operators and correctable errors made possible by non-unitary quantum channels that makes universality possible, not any of these more specific requirements that are common in universal schemes.

\begin{figure*}\label{Fig:ADJ}
\includegraphics{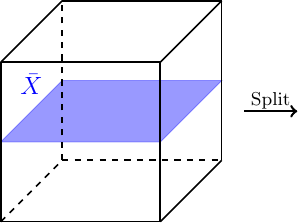}
\includegraphics{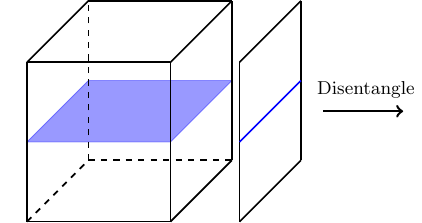}
\includegraphics{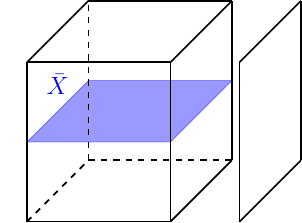}
\includegraphics{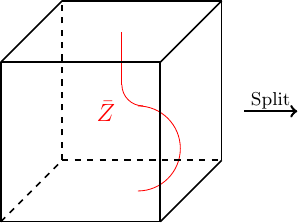}
\includegraphics{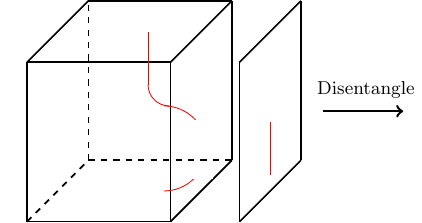}
\includegraphics{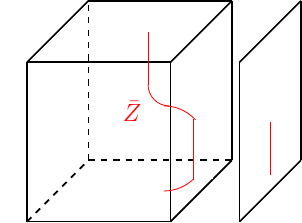}
\caption{Representation of action on logical operators of removing a layer during dimensional conversion between 2D and 3D surface codes. A boundary is viewed as a 2D surface code (split) and then a locality-preserving logical circuit that is a product of nearest-neighbour CNOT operators is applied between it and the face of the 3D code adjacent to it to disentangle it from the 3D code. The final result is that the original logical operators are mapped to the corresponding logical operators on the new code with the layer removed. (The removed layer is discarded, so its state is not important.) A layer can be added by the inverse process. \label{Fig:ADJ}}
\end{figure*}

\subsection{Pieceable Fault Tolerance with Delayed Error Correction}\label{AppendixD3}
Code conversion uses code deformation to fault-tolerantly (and non-unitarily) switch between two codes that together admit a universal set of bounded-spread logical operator implementations.  The gradational technique of pieceable fault tolerance can also allow direct implementation of an unbounded-spread logical operator implementation on a single code \cite{Knill2,Yoder2}. Indeed, the scheme of Ref.~\cite{Brown} for implementing a non-Clifford logical operator on the two-dimensional surface code can be understood as taking such an approach, supplemented with the additional technique of delayed error correction to address the problem of uncorrected errors accumulating over time. We here describe a perspective on the scheme, focussing on elements salient to this understanding. For a complete presentation of the scheme, we refer the reader to Ref.~\cite{Scruby2}.

Specifically, Ref.~\cite{Brown} implements an unbounded-spread $\overline{\text{CCZ}}$ logical operator on the two-dimensional surface code. As with all logical operators there exists a unitary circuit that implements this logical operator, but this circuit requires entangling each physical qubit non-locally with a large number of other physical qubits, allowing for errors to spread over these qubits. This problem is addressed in the way discussed in Sec.~\ref{Sec:IVA1} -- the circuit is decomposed into an extensive number of bounded-spread (in this case transversal) steps and an operation performed between each step to allow for error-correction in order to give a gradational implementation.

In principle, this gradational implementation could be performed by standard code deformation. At intermediate times, the encoded state is not a logical state of the two-dimensional surface code, but will be a logical state of some more general code on the physical qubits (i.e.~that with logical space stabilised by the image of the surface code stabilisers under the action of the physical CCZ operators performed up to that point). Error correction at intermediate times could then be performed by implementing any appropriate decoders on these intermediate codes. However, while theoretically possible, this is practically difficult since the intermediate code does not retain important structural elements that make decoding the surface code feasible. Specifically, the ``stabilisers'' of the intermediate code are non-Pauli operators (since the physical operators applied are non-Clifford), meaning that these codes are not true stabiliser codes. Moreover, these ``stabilisers'' are non-local and high weight (since the physical operators are non-local, entangling operators), meaning that the intermediate codes are not topological or even LDPC. Thus, syndrome extraction for standard decoding of the intermediate codes would require high-weight, non-local, non-Pauli measurements, which are unlikely to be practical.

The scheme avoids this problem by taking a more nuanced approach based on pieceable fault tolerance to prevent uncontrolled arising of errors without standard decoding on intermediate codes. To clarify how this is done, it is helpful to consider $X$ and $Z$ type errors separately, which can be done because the surface code is CSS.

To address $X$ errors, we note that the $Z$-type stabilisers (which detect such errors) commute with all the physical CCZ operators. This means that the $Z$-type stabilisers of all intermediate codes during the code deformation are the same as in the two-dimensional surface code. Thus, $X$ errors can in principle be corrected by use of a standard surface code decoder. It is important that this correction is performed because uncorrected $X$ errors would be spread by subsequent CCZ operators, and thus become high-weight, uncorrectable errors. In practice, it is important to the scheme that this correction is done in constant time, which requires use of just-in-time decoding \cite{Bombin4} -- a technique that allows for the single-shot error-correction of $X$ errors in the three-dimensional surface code to be leveraged in the context of this (2+1)D scheme. However, the fundamental approach remains to correct $X$ errors by measuring $Z$ stabilisers and inferring syndromes, as with standard decoding.

To address $Z$ errors, the important observation is that $Z$ errors commute with all the physical CCZ operators and so are not spread by them. This means that correction of such errors can be delayed without concern of spreading. Indeed, this is also used in Ref.~\cite{Yoder2}, where such errors are indeed left unaddressed until after the deformation is complete and a standard decoder can again be used on the original code. However, this delaying also allows for new $Z$ errors to arise unhindered throughout the duration of the code deformation. Since this duration scales polynomially with the code size, $l$, increasing the code size cannot be assumed to provide improved protection against such errors. Thus, entirely delaying dealing with $Z$ errors will not ensure maintenance of a threshold and hence will not be fault-tolerant by our definition.

In order to address this issue, Ref.~\cite{Brown} extends upon pieceable fault tolerance by using delayed error correction. Specifically, syndrome information about $Z$ errors is extracted at each step of the deformation by qubitwise $X$ measurements. In particular, a teleportation circuit is applied at the level of physical qubits. (This process involves entangling two layers of 2D surface codes, and so is interpreted in \cite{Brown} as creating an $(l\times l \times 2)$ 3D surface code.) This teleportation circuit can be understood as similar to teleportation error-correction (as described in Sec.~\ref{Sec:IVB}) but performed on the level of physical qubits instead of at the logical level because the logical $X$ operator of the intermediate codes cannot practically be measured. The $X$ measurements performed in this circuit yield information that can be used to reconstruct the state of $X$ stabilisers of a three-dimensional surface code equvialent to the (2+1)D surface code realised during code deformation. (Repeated measurements are not necessary because the measurements are qubitwise instead of being of whole stabilisers, and so measurement errors can be treated equivalently to physical errors \cite{Brown}.) Since $Z$ errors of a 3D surface code cannot be corrected locally, correction of the errors implied by this syndrome {cannot be completed} until the end of the code deformation process when all measurements have been made. Nevertheless, the syndrome of $Z$ errors on the final 2D code can be inferred from this 3D code syndrome and so $Z$ errors can be appropriately corrected. The $Z$ errors of the 3D code correspond to all $Z$ errors occuring throughout the deformation and so this correction ensures fault tolerance.

In summary, the scheme of Ref.~\cite{Brown} may thus be understood as a gradational logical operator implementation of the form $\mathcal{U}=\prod_{i=1}^t \mathcal{K}_i\mathcal{U}_i$. Here, each $\mathcal{U}_i$ is a transversal implementation of physical CCZ gates such that $\prod_i \mathcal{U}_i=\mathcal{U}$. Each $\mathcal{K}_i$ for $i<t$ can be expressed as $\mathcal{K}_i=\mathcal{D}^{(x)} \mathcal{T}$, where $\mathcal{D}^{(x)}$ is a (just-in-time) decoder that corrects $X$ errors and $\mathcal{T}$ is a physical teleportation circuit that extracts syndrome information about $Z$ errors. The final channel $\mathcal{K}_t$ also includes correction of $Z$ errors based on the collected syndrome information, as described in the previous paragraph. Thus, the scheme can be understood as a form of gradational logical operator implementation.

\section{ Non-ELNP Fault-Tolerant Quantum Channels}\label{AppendixE}
Throughout this paper, we have { focussed on quantum channels whose fault tolerance is ensured by standard threshold theorems -- first LNP channels, and then more generally ELNP channels. This approach allows us to characterise general classes of quantum channels as fault-tolerant without requiring the proof of tailored threshold theorems for each channel, or requiring the use of tailored decoders.} It follows the lead of most previous approaches to fault tolerance and so is appropriate for contextualising most of these approaches, as is evident from Sec.~\ref{Sec:IV} and Appendix \ref{AppendixD}.

{ Nevertheless, an alternative approach is possible and can allow for new approaches to implementing a universal set of logical operators. Specifically, recall from Sec.~\ref{Sec:IIA} that a quantum channel, $\mathcal{A}$ is fault-tolerant for a particular stochastic noise model, $\mathcal{E}$, provided that the code family has a threshold for the effective noise channel, $\mathcal{E}_{\mathcal{A}}$, i.e.~there exists a decoder $\mathcal{D}$ and threshold $\eta_{\mathcal{A}}>0$ such that
\begin{equation}\label{EqFT2F}
  \lim_{l\to\infty} \left\|\mathcal{D}\circ\mathcal{E}_{\mathcal{A}_\eta}-\mathcal{I}\right\|_{\diamond,\mathcal{L}}=0\,, \quad \text{for}\ \eta<\eta_{\mathcal{A}}\,.
\end{equation}
In particular settings, it can be possible to prove that this is true even if $\mathcal{E}_{\mathcal{A}}$ is not a local stochastic noise model. This can allow $\mathcal{A}$ to be fault-tolerant even if it is not ELNP. In particular, this could allow for a universal set of logical operators implemented by unitary, fault-tolerant quantum channels.} This possibility appears to have received relatively little attention to date; while tailored decoders have been used within broader schemes to address specific problems (such as just-in-time decoding \cite{Bombin4,Brown}), they have rarely been used as the single technique to circumvent no-go theorems on fault-tolerant gate sets. Thus, further exploration could be fruitful.

However, we note one prominent example of this approach that has been proposed by Jochym O'Connor and Laflamme \cite{JochymOConnor2}. Specifically, this scheme is based on alternated concatenation of Steane and Reed-Muller codes (for a total of $\frac{l}{2}$ of each type of code). A logical $\bar{T}$ operator is transversal on the Reed-Muller code, but the standard unitary circuit has a spread of three on the Steane code. This means that the unitary implementation of $\bar{T}$ on the concatenated code by applying the appropriate implementation at each level of concatenation has a spread of $3^{\frac{l}{2}}$ since a single-qubit error can be spread by a factor of $3$ for each layer of Steane code. Similarly, $\bar{H}$ is transversal on Steane codes but not Reed-Muller codes, and so similarly has spread that scales with the number of layers of Reed-Muller codes, $\frac{l}{2}$. Hence, the implementations have unbounded-spread and so (by Theorem \ref{Th1}) cannot be LNP. Indeed, since they are unitary, it also follows that they cannot be ELNP (as shown in Appendix \ref{AppendixA}).

Nevertheless, these unbounded-spread, unitary implementations of logical operators are fault-tolerant (i.e.~have a threshold) with standard concatenated code decoders and a depolarising noise model \cite{Chamberland}. This is possible because large errors that arise due to spreading on the layers on which a logical operator implementation is not transversal can be corrected using the other type of layer, with respect to which errors are not spread such that uncorrectable errors do not arise as a result of the implementation. Thus, if $\mathcal{E}$ is a depolarising channel and $\mathcal{H}$ and $\mathcal{T}$ are implementations of $\bar{H}$ and $\bar{T}$ respectively as described then $\mathcal{E}_{\mathcal{H}}$ and $\mathcal{E}_{\mathcal{T}}$ are noise models that are non-local, but which have specific structure that arises from the deliberate construction of the code family and channels that nonetheless allows for a threshold.

We note one illuminating nuance of the scheme that prevents it from allowing for arbitary quantum gates to be implemented fault-tolerantly by unitary channels. Specifically, $\mathcal{H}$ is not fault-tolerant for noise model $\mathcal{E}_{\mathcal{T}}$ and $\mathcal{T}$ is not fault-tolerant for noise model $\mathcal{E}_{\mathcal{H}}$, since the resulting noise models ($\mathcal{E}_{\mathcal{H}\mathcal{T}}$ and $\mathcal{E}_{\mathcal{T}\mathcal{H}})$ have non-vanishing probability of uncorrectable errors affecting both types of codes in the concatenation. This means that composite logical operators, $\bar{H}\bar{T}$ and $\bar{T}\bar{H}$ cannot be implemented fault-tolerantly by unitary channels; intermediate error correction is required between the implementations of $\bar{H}$ and $\bar{T}$. Thus, while a universal gate set (i.e.~generators of a group dense in the group of all quantum gates) can be implemented fault-tolerantly by unitary channels, the composition of such gates to realise arbitary operations requires non-unitary channels. By contrast, since LNP channels map all local stochastic noise models to local stochastic noise models, any sequence of them of length constant in the code size can be applied while retaining fault tolerance.

In summary, while the scheme of Ref.~\cite{JochymOConnor2} is ostensibly similar to intrinsically fault-tolerant schemes using Steane and Reed-Muller codes \cite{Paetznik,Anderson}, it actually uses a fundamentally different technique to leverage the properties of these codes from code switching or gauge fixing. Specifically, this technique is based on designing a code family and quantum channels specifically to ensure fault tolerance, even without preserving the locality of noise (even up to errors of probability zero) In addition to concatenation, an alternative version of this approach based on the homological product of codes has also been proposed \cite{JochymOConnor4}. It may be fruitful to further pursue analogous approaches on more general code families.

\end{document}